\documentclass[reqno,12pt]{amsart}

\usepackage{amsmath}
\usepackage{latexsym}
\usepackage{amssymb}
\usepackage{hyperref}
\usepackage{graphicx}
\usepackage{epstopdf}
\usepackage{ifpdf}
\ifpdf
\fi
\usepackage{cite}

\makeatletter
 
 \newcommand{\Rmnum}[1]{\expandafter\@slowromancap\romannumeral #1@}
 \makeatother

\newtheorem{theorem}{Theorem}[section]
\newtheorem{definition}{Definition}[section]
\newtheorem{proposition}[theorem]{Proposition}
\newtheorem{lemma}[theorem]{Lemma}

\newtheorem{remark}[theorem]{Remark}
\newtheorem{assumption}[theorem]{Assumption}

\newcommand{\R}{{\mathbb R}}

\newcommand{\C}{{\mathbb C}}

\newcommand{\be}{\begin{equation}}
\newcommand{\ee}{\end{equation}}
\newcommand{\bea}{\begin{eqnarray}}
\newcommand{\eea}{\end{eqnarray}}
\newcommand{\ba}{\begin{array}}
\newcommand{\ea}{\end{array}}

\newcommand{\ol}{\overline}

\newcommand{\id}{\mathbb{I}}

\newcommand{\re}{\mathrm{Re}}
\newcommand{\im}{\mathrm{Im}}

\newcommand{\eps}{\varepsilon}

\newcommand{\sig}{\sigma}
\newcommand{\Sig}{\Sigma}
\newcommand{\lam}{\lambda}

\newcommand{\gam}{\gamma}
\newcommand{\Gam}{\Gamma}
\newcommand{\om}{\omega}

\newcommand{\dta}{\delta}
\newcommand{\Dta}{\Delta}

\newcommand{\tha}{\theta}

\numberwithin{equation}{section}

\begin{document}
\title[RHP for the short pulse equation]{Long-time asymptotics for the short pulse equation}

\author[J.Xu]{Jian Xu*}
\address{College of Science\\
University of Shanghai for Science and Technology\\
Shanghai 200093\\
People's  Republic of China}
\email{corresponding author: jianxu@usst.edu.cn}



\date{\today}

\begin{abstract}
In this paper, we analyze the long-time behavior of the solution of the initial value problem (IVP) for the short pulse (SP) equation. As the SP equation is a complete integrable system, which posses a Wadati-Konno-Ichikawa (WKI)-type Lax pair, we formulate a $2\times 2$ matrix Riemann-Hilbert problem to this IVP by using the inverse scattering method. Since the spectral variable $k$ is the same order in the WKI-type Lax pair, we construct the solution of this IVP parametrically in the new scale $(y,t)$, whereas the original scale $(x,t)$ is given in terms of functions in the new scale, in terms of the solution of this Riemann-Hilbert problem. However, by employing the nonlinear steepest descent method of Deift and Zhou for oscillatory Riemann-Hilbert problem, we can get the explicit leading order asymptotic of the solution of the short pulse equation in the original scale $(x,t)$ as time $t$ goes to infinity. 

\end{abstract}

\maketitle

\section{Introduction}
The present work is devoted to the study of the long-time asymptotic behavior of the short pulse (SP) equation formulated on the whole line,
\begin{subequations}
\be\label{spe}
u_{xt}=u+\frac{1}{6}(u^3)_{xx},
\ee
where $u(x,t)$ is a real-valued function, which represents the magnitude of the electric
field, while the subscripts $t$ and $x$ denote partial differentiations, with the initial value data
\be\label{spe-ini}
u(x,t=0)=u_0(x),\quad x\in \R,
\ee
and assuming that $u_0(x)$ lies in Schwartz space.
\end{subequations}

\par
The SP equation was reproposed in \cite{swpd} by Sch\"afer and Wayne to describe the propagation of ultra-short optical
pulses in silica optical fibers. Usually, in nonlinear optic, the nonlinear Schr\"odinger (NLS) equation was always used to model the slowly varying wave trains. As the pulse duration shortens, however, the NLS equation becomes less accurate, the SP equation provides an increasingly better
approximation to the corresponding solution of the Maxwell equations \cite{cjsw}. For the details of physical background, see \cite{swpd} and references therein.
\par
Actually, the SP equation appeared first as one of Rabelo's equations which describe pseudospherical surfaces, possessing a zero-curvature representation, in \cite{rspam}. Recently, the Wadati-Konno-Ichikawa (WKI) type Lax pair of the SP equation was rediscovered in \cite{ssjpsj} (see the following (\ref{Lax-x})).  The integrable properties of SP equation like bi-Hamiltonian structure and the conservation laws were studied in \cite{bjmp,bpla}. The loop-soliton
solutions the short pulse equation was found in \cite{ssjpa}. The
connection between the short pulse equation and the sine-Gordon equation through the hodograph transformation was found by Matsuno, and thus, multi-soliton solutions including multi-loop
and multi-breather ones were given in \cite{mjpsj}. And a lot of generalizations of the SP equation, such as vector SP equation, discretizations of SP equation, complex SP equation,and so on, were studied in \cite{sjpsj,fjpa,fpd} and references therein.
\par
The local well-posedness in $H^2$ (which denotes the usually Sobelev space) and non-existence of smooth traveling wave solutions were shown in \cite{swpd}, and global well-posedness of small solutions was proved in \cite{dcpde} for SP equation in $H^2$ by using conservation laws. In \cite{ldpde}, Liu, Pelinovsky and Sakovich showed the blow-up result for the SP equation for large data.

\par
The purpose of this paper is to analyzing the long-time asymptotic behavior of the SP equation.
Due to the SP equation admits a Lax pair, the inverse scattering transform method can be used to solve the initial value problem for the SP equation. Here, we relate the inverse scattering problem to a $2\times 2$-matrix Riemann-Hilbert problem. The most important advantage of formulating the initial value problem (\ref{spe})-(\ref{spe-ini}) as a Riemann-Hilbert problem is that it can be analyzed the long-time asymptotic behavior of the solution of the initial value problem by employing the nonlinear steepest descent method introduced by Deift and Zhou \cite{dz}.
This method has previous applied to many integrable equations, such as the NLS equation \cite{diz}, the Sine-Gordon equation \cite{cvz}, the KdV equation \cite{gt}, the Fokas-Lenells equation \cite{jfjde}, the Camassa-Holm equation \cite{amkst} and so on.

\par
Although the long-time asymptotic analysis of (\ref{spe}) is in many ways similar to those of integrable equations, it also presents some distinctive features: {\em (1)} In order to arrive at a Riemann-Hilbert problem with the appropriate boundary condition at infinity, {\bf since the spectral variable $k$ is the same order in the Lax pair} (which is different from the Fokas-Lenells equation and the Camassa-Holm equation), it first has to be transformed by introduction of a matrix $G(x,t)$. This modification is made possible since the conservation law (\ref{conslaw}) holds. {\em (2)} In order to construct the solution $u(x,t)$, we need the asymptotic information as spectral variable $k\rightarrow 0$, not as $k\rightarrow \infty$. Hence, we need analyze the spectral problem near both $k=\infty$ and $k=0$, respectively. Although the asymptotic behavior as $k\rightarrow 0$ contains all the information of the solution $u(x,t)$, we also cannot use it directly since the exponential term $e^{ikp(x,t,k)\sig_3}$ of the jump matrix $J(x,t,k)$ involves the unknown function $m(x,t)$, see (\ref{mfundef}). We need introducing a new scale variable $y(x,t)$, this leads to the solution $u$ reconstructed from the Riemann-Hilbert problem becomes implicitly ($u$ depends on $(y,t)$ and $y$ depends on $(x,t)$). Fortunately, we can obtain the leading order explicit asymptotic behavior of the solution $u$ which depends on the original variable $(x,t)$, since the reflection coefficient $r(k)$ is rapidly decaying as $k\rightarrow 0$. {\bf The solution $u(x,t)$ is constructed from a $2\times 2$-matrix Riemann-Hilbert problem in terms of the order $O(k)$ as $k\rightarrow 0$}, {\em which is different from the Camassa-Holm equation \cite{amkst} and the modified Hunter-Saxton equation in \cite{amszip}, where they construct the solution from a $1\times 2$-vector Riemann-Hilbert problem.}

\begin{remark}
\par
We note that although there exists a gauge transformation between WKI type Lax pair and AKNS type Lax pair, (especially, the short pulse equation can be transformed to Sine-Gordon equation \cite{ssjpa}), it can not use this gauge transformation to analyze the long-time asymptotic behavior of the short pulse, since there is an unknown function related to the solution $u(x,t)$ in the transformation. We need analyze the asymptotic behavior directly. By formulating a vector Riemann-Hilbert problem in this paper, we successfully analyse the asymptotic behavior as $t\rightarrow \infty$. To the author's knowledge, this is the first time to analyze long-time asymptotic behavior of the nonlinear evolution equation with WKI type Lax pair.
In fact, many equations with physical implications have been found to posses WKI type Lax pair, such as the stretched rope equation \cite{iwkjpsj}, the modified Camassa-Holm (mCH) equation \cite{qltmp}, and so on. 
Hence, it may make us to analyze the long-time behavior of the mCH equation in the future work.
\end{remark}

The main results of this paper are summarized by the following theorems:

\begin{theorem}\label{main1}
Let $u_0(x)$ satisfy the initial value (\ref{spe-ini}) and be such that no discrete spectrum is present. For $\xi=\frac{x}{t}>\eps$, $\eps$ be any small positive number, the solution $u(x,t)$ of the initial value problem (\ref{spe})-(\ref{spe-ini}) tends to $0$ fast decay as $t\rightarrow \infty$.
\end{theorem}

\begin{theorem}\label{main2}
Let $u_0(x)$ satisfy the hypotheses of Theorem \ref{main1}. For $\xi=\frac{x}{t}<-\eps$, $\eps$ be any small positive number, the solution $u(x,t)$ of the initial value problem (\ref{spe})-(\ref{spe-ini}) equals
\be\label{uxtasy}
u(x,t)=\sqrt{\frac{-4\nu(\kappa_0)}{\kappa_0 t}}\sin{\{\frac{t}{\kappa_0}+\nu(\kappa_0)\ln{(\frac{4t}{\kappa_0})}+\phi(\kappa_0)\}}+O\left(\frac{\ln{(t)}}{t}\right),\quad \mbox{as }t\rightarrow \infty,
\ee
where
\be\label{kappa0def}
\kappa_0=\frac{1}{\sqrt{-4\xi}}
\ee
and
$\nu(\kappa_0)$ is defined as (\ref{nudef}) replaced $k_0$ with $\kappa_0$, $\phi(\kappa_0)$ is defined as (\ref{phikappadef}).
\end{theorem}

\begin{remark}
The sectors of different asymptotic behavior match, as $\eps \rightarrow 0$, through the fast decay. Indeed, as $\frac{x}{t}\rightarrow 0^-$, then $\kappa_0\rightarrow \infty$ and $\nu(\kappa_0)\rightarrow 0$ and thus the amplitude in (\ref{uxtasy}) decays faster.
\end{remark}
\par
{\bf Organization of the paper:} In section 2, since the associated Lax pair of short pulse (\ref{spe}) has singularities at $k=0$ and $k=\infty$, we perform the spectral analysis to deal with the two singularities, respectively. In section 3, we formulate the associated vector Riemann-Hilbert in an alternative space variable $y$ instead of the original space variable $x$. Hence, we can reconstruct the solution $u(x,t)$ parameterized from the solution of the Riemann-Hilbert problem via the asymptotic behavior of the spectral variable at $k=0$. In section 4 and 5, we can also obtain the asymptotic relation between $y$ and $x$ when analyzing the vector Riemann-Hilbert problem by using the nonlinear steepest descent method. Hence, we can calculate the leading order asymptotic behavior of the solution $u(x,t)$ and prove the main results of this paper, i.e., theorem \ref{main1} and \ref{main2}, respectively.

\section{Spectral Analysis}
The beginning of the long-time asymptotic analysis is to formulate the initial value problem for the short pulse equation to a Riemann-Hilbert problem. It depends on short pulse equation admits a WKI type Lax pair,
\begin{subequations}\label{Laxpair}
\be\label{Lax-x}
\Psi_x=U(x,t,k)\Psi,
\ee
\be\label{Lax-t}
\Psi_t=V(x,t,k)\Psi.
\ee
\end{subequations}
where
\be\label{Udef}
U=ik U_1.
\ee
and
\be\label{Vdef}
V=\frac{ik}{2}u^2 U_1+\frac{1}{4ik}\sig_3-\frac{iu}{2}\sig_2,
\ee
with
\be\label{U1def}
U_1=\left(\ba{cc}1&u_x\\u_x&-1\ea\right),\sig_3=\left(\ba{cc}1&0\\0&-1\ea\right),\sig_2=\left(\ba{cc}0&-i\\i&0\ea\right).
\ee
Usually, we only use the $x-$part of Lax pair to analyze the initial value problem for the integrable equation by inverse scattering transform method. The $t-$part of Lax pair is only used to determine the time evolution of the scattering data. However,
from the Lax pair (\ref{Laxpair}), we know that there are singularities at $k=\infty$ and $k=0$. We notice that the Camassa-Holm equation \cite{amkst}, modified Hunter-Saxton equation \cite{amszip} also possess two singularities in the Lax pair of integrable nonlinear evolution equations.
In order to reconstruct the solution $u(x,t)$ of the SP equation (\ref{spe}), we need use the $t-$part or using the expansion of the eigenfunction as spectral parameter $k\rightarrow 0$. Hence, in the following we use two different transformations to analyze these two singularities ($k=\infty$ and $k=0$), respectively.

\subsection{For $k=0$}

\subsubsection{The closed one-form}
Introducing the following transformation
\be
\Psi(x,t,k)=\mu^0(x,t,k)e^{(ikx+\frac{t}{4ik})\sig_3},
\ee
then we get the Lax pair of $\mu^0$
\be\label{mu0Laxe}
\left\{
\ba{l}
\mu^0_x-ik[\sig_3,\mu^0]=V_1^0\mu^0,\\
\mu^0_t-\frac{1}{4ik}[\sig_3,\mu^0]=V_2^0\mu^0,
\ea
\right.
\ee
where
\be
V_1^0=\left(\ba{cc}0&iku_x\\iku_x&0\ea\right),\quad V_2^0=\left(\ba{cc}\frac{ik}{2}u^2&\frac{ik}{2}u^2u_x-\frac{u}{2}\\\frac{ik}{2}u^2u_x+\frac{u}{2}&-\frac{ik}{2}u^2\ea\right)
\ee

Letting $\hat A$ denotes the operators which acts on a $2\times 2$ matrix $X$ by $\hat A X=[A,X]$ , then the Lax pair of $\mu^0$ (\ref{mu0Laxe}) can be written as
\be
d(e^{-(ikx+\frac{t}{4ik})\hat \sig_3}\mu^0)=W^0(x,t,k),
\ee
where $W^0(x,t,k)$ is the closed one-form defined by
\be
W^0(x,t,k)=e^{-(ikx+\frac{t}{4ik})\hat \sig_3}(V_1^0dx+V_2^0dt)\mu^0.
\ee

\subsubsection{The Jost functions $\mu^0_j$}

We define two eigenfunctions $\{\mu^0_j\}_{j=1}^2$ of (\ref{mu0Laxe}) by the Volterra integral equations,
\begin{subequations}\label{mu0jdef}
\be\label{mu01def}
\mu^0_1(x,t,k)=\id+\int_{-\infty}^{x}e^{ik(x-y)\hat\sig_3}V^0_1(y,t,k)\mu^0_1(y,t,k)dy,
\ee
\be\label{mu02def}
\mu^0_2(x,t,k)=\id-\int_{x}^{+\infty}e^{ik(x-y)\hat\sig_3}V^0_1(y,t,k)\mu^0_2(y,t,k)dy.
\ee
\end{subequations}

\begin{proposition}(Analytic property)
From the above definition, we find that the functions $\{\mu^0_j\}_1^2$ are bounded and analytic properties as following:
\begin{itemize}
 \item $[\mu^0_1]_1(x,t,k)$ is bounded and analytic in $D_2$, $[\mu^0_1]_2(x,t,k)$ is in $D_1$;
 \item $[\mu^0_2]_1(x,t,k)$ is bounded and analytic in $D_1$, $[\mu^0_2]_2(x,t,k)$ is in $D_2$.
\end{itemize}
where $[\mu_j]_i$ denotes the $i-$th column of $\mu_j$, $D_1$ denotes the upper-half plane and $D_2$ denotes the lower-half plane of the complex $k-$sphere.
\end{proposition}
\begin{proposition}(Asymptotic property)
The functions $\mu^0_j(x,t,k)$ have the expansions in powers of $k$, for $k\rightarrow 0$,
\be
\mu^0_j(x,t,k)=\id+iu(x,t)\sig_1k+[-\frac{u^2}{2}\id+i(u^2u_x-2u_t)\sig_2]k^2+O(k^3).
\ee
\end{proposition}

\subsection{For $k=\infty$}

\subsubsection{The closed one-form}
Define a $2\times 2$ matrix-value function $G(x,t)$ as
\be\label{Gdef}
G(x,t)=\sqrt{\frac{\sqrt{m}+1}{2\sqrt{m}}}\left(\ba{cc}1&-\frac{\sqrt{m}-1}{u_x}\\ \frac{\sqrt{m}-1}{u_x}&1\ea\right),
\ee
where $m$ is a function of $(x,t)$ defined by
\be\label{mfundef}
m=1+u_x^2.
\ee
\begin{remark}
Notice that when $u_x\rightarrow 0$, the nominator $\sqrt{m}-1$ is a high order infinitesimal than denominator $u_x$. So, the matrix function $G(x,t)$ is well-defined.
\end{remark}
Define
\be
p(x,t,k)=x-\int_{x}^{\infty}(\sqrt{m(x',t)}-1)dx'-\frac{t}{4k^2}.
\ee
As we can write the SPE (\ref{spe}) into the conservation law form:
\be\label{conslaw}
(\sqrt{m})_t=\frac{1}{2}(u^2\sqrt{m})_x,\quad m=1+u^2_x,
\ee
we get
\be
p_x=\sqrt{m},\quad p_t=\frac{1}{2}u^2\sqrt{m}-\frac{1}{4k^2}.
\ee
And introducing a transformation
\be
\Psi(x,t,k)=G(x,t)\mu(x,t,k)e^{ikp(x,t,k)\sig_3}
\ee
then we find the Lax pair equations
\be\label{muLaxe}
\left\{
\ba{l}
\mu_x-ik p_x[\sig_3,\mu]=V_1\mu,\\
\mu_t-ik p_t[\sig_3,\mu]=V_2\mu,
\ea
\right.
\ee
where
\begin{subequations}
\be
V_1=\frac{iu_{xx}}{2m}\sig_2,
\ee
\be
V_2=\frac{1}{4ik}(\frac{1}{\sqrt{m}}-1)\sig_3+\frac{iu^2u_{xx}}{4m}\sig_2-\frac{1}{4ik}\frac{u_x}{\sqrt{m}}\sig_1,
\ee
with $\sig_1=\left(\ba{cc}0&1\\1&0\ea\right)$.
\end{subequations}

 then the equations in (\ref{muLaxe}) can be written in differential form as
\be\label{mudiffform}
d(e^{-ikp(x,t,k)\hat\sig_3}\mu)=W(x,t,k),
\ee
where $W(x,t,k)$ is the closed one-form defined by
\be\label{Wdef}
W=e^{-ikp(x,t,k)\hat\sig_3}(V_1dx+V_2dt)\mu.
\ee

\subsubsection{The Jost functions $\mu_j$}
We define two eigenfunctions $\{\mu_j\}_1^2$ of (\ref{muLaxe}) by the Volterra integral equations
\begin{subequations}\label{mujdef}
\be\label{mu1def}
\mu_1(x,t,k)=\id+\int_{-\infty}^{x}e^{ik[p(x,t,k)-p(y,t,k)]\hat\sig_3}V_1(y,t,k)\mu_1(y,t,k)dy,
\ee
\be\label{mu2def}
\mu_2(x,t,k)=\id-\int_{x}^{+\infty}e^{ik[p(x,t,k)-p(y,t,k)]\hat\sig_3}V_1(y,t,k)\mu_2(y,t,k)dy.
\ee
\end{subequations}

\begin{proposition}(Analytic property)
From the above definition, we find that the functions $\{\mu_j\}_1^2$ are bounded and analytic properties as following:
\begin{itemize}
 \item  $[\mu_1]_1(x,t,k)$ is bounded and analytic in $D_2$, $[\mu_1]_2(x,t,k)$ is in $D_1$;
 \item $[\mu_2]_1(x,t,k)$ is bounded and analytic in $D_1$, $[\mu_2]_1(x,t,k)$ is in $D_2$.
\end{itemize}
\end{proposition}
\begin{proposition}(Large $k$ property)
The matrix functions $\mu_j(x,t,k)$ also satisfy the asymptotic condition
\be\label{Masy}
\mu_j(x,t,k)=\id+\frac{D_1(x,t)}{k}+O(\frac{1}{k^2}),\quad k\rightarrow \infty,
\ee
where $\id$ is an $2\times 2$ identity matrix, and the off-diagonal entries of the matrix $D_1(x,t)$ are
\be
D_{12}(x,t)=\frac{i}{4}\frac{u_{xx}}{m\sqrt{m}},\quad D_{21}(x,t)=\frac{i}{4}\frac{u_{xx}}{m\sqrt{m}}.
\ee
\end{proposition}

\subsubsection{The scattering matrix $S(k)$}
Because the eigenfunctions $\mu_1(x,t,k)$ and $\mu_2(x,t,k)$ are both the solutions of the Lax pair (\ref{muLaxe}), they are related by a matrix $S(k)$ which is independent of the variable $(x,t)$.
\be\label{scatermatrix}
\mu_1(x,t,k)=\mu_2(x,t,k)e^{ikp(x,t,k)\hat\sig_3}S(k).
\ee
By the definition of $\mu_j(x,t,k),j=1,2$ (\ref{mujdef}), the matrix $S(k)$ has the form
\be\label{Skdef}
S(k)=\left(\ba{cc}\ol{a(\bar k)}&b(k)\\-\ol{b(\bar k)}&a(k)\ea\right).
\ee
The function $a(k)$ can be computed by
\be\label{akdef}
a(k)=\det{([\mu_2]_1,[\mu_1]_2)},
\ee
where $\det{(A)}$ means the determinate of a matrix $A$. We can know that $a(k)$ is analytic in $D_1$.
\begin{assumption}
In this paper, we assume that the initial value $u_0(x)$ is chosen such that $a(k)$ has no zero (usually, we can assume $u_0(x)$ has small norm).
\end{assumption}
\begin{proposition}
The proposition (\ref{promujrelmu0j}) together with (\ref{akdef}) allows expressing the expansions in powers of $k$ of $a(k)$ at $k=0$,
\be\label{akasy}
a(k)=1+ikc-\frac{c^2}{2}k^2+O(k^3).
\ee
\end{proposition}

\subsection{The relation between $\mu_j(x,t,k)$ and $\mu^0_j(x,t,k)$}
As usual, we use the eigenfunctions $\mu_j$ to define the matrix $M(x,t,k)$ (see (\ref{Mdef})) which is used to formulate a Riemann-Hilbert problem. However, in order to construct the solution $u(x,t)$ from the associate Riemann-Hilbert problem, we need the asymptotic behavior of $\mu_j$ as $k\rightarrow 0$. So, we need related the eigenfunctions $\mu_j(x,t,k)$ to $\mu^0_j(x,t,k)$.
\par
Note that the eigenfunctions $\mu(x,t,k)$ and $\mu^0(x,t,k)$ being related to the same Lax pair (\ref{Laxpair}), must be related to each other as
\be\label{mujrelmu0j}
\mu_j(x,t,k)=G^{-1}(x,t)\mu^0_j(x,t,k)e^{(ikx+\frac{t}{4ik})\sig_3}C_j(k)e^{-ikp(x,t,k)\sig_3}
\ee
with $C_j(k)$ independent of $x$ and $t$. Evaluating (\ref{mujrelmu0j}) as $x\rightarrow \pm\infty$ gives
\be
C_1(k)=e^{-ikc\sig_3},\quad \quad C_2(k)=\id,
\ee
where $c=\int_{-\infty}^{+\infty}(\sqrt{m(x,t)}-1)dx$ is a quantity conserved under the dynamics governed by (\ref{spe}).

\begin{proposition}\label{promujrelmu0j}
The functions $\mu_j(x,t,k)$ and $\mu^0_j(x,t,k)$ are related as follows:
\begin{subequations}
\be
\mu_1(x,t,k)=G^{-1}(x,t)\mu^0_1(x,t,k)e^{-ik\int_{-\infty}^x(\sqrt{m(x',t)}-1)dx'\sig_3},
\ee
\be
\mu_2(x,t,k)=G^{-1}(x,t)\mu^0_2(x,t,k)e^{ik\int_{x}^{+\infty}(\sqrt{m(x',t)}-1)dx'\sig_3}.
\ee
\end{subequations}
\end{proposition}

%

\section{The Riemann-Hilbert problem}

\subsection{Definition of a Riemann-Hilbert problem}
In this subsection, we first explain what a Riemann-Hilbert problem is:
\begin{definition}
Let the contour $\Gam$ be the union of a finite number of smooth and oriented curves (orientation means that each arc of $\Gam$ has a positive side and a negative side: the positive (respectively, negative) side lies to the left (respectively, right) as one traverses the contour in the direction of the arrow) on the Riemann sphere $\bar \C$ (i.e. the complex plane with the point at infinity) such that $\bar \C\backslash \Gam$ has only a finite number of connected components. Let $V(k)$ be an $2\times 2$ matrix defined on the contour $\Gam$. The Riemann-Hilbert problem $(\Gam,V)$  is the problem of finding an $2\times 2$ matrix-valued function $M(k)$ that satisfies
\begin{enumerate}
\item $M(k)$ is analytic for $k\in \bar \C\backslash \Gam$ and extends continuously to the contour $\Gam$.
\item $M_+(k)=M_-(k)V(k),\quad k\in\Gam$.
\item $M(k)\rightarrow \id,\quad as \quad k\rightarrow \infty.$
\end{enumerate}
\end{definition}
\par
The Riemann-Hilbert problem can be solved as follows (see, \cite{bc}). Assume that $V(k)$ admits some factorization
\be\label{BCRHPfac}
V(k)=b_-^{-1}(k)b_+(k),
\ee
where
\be\label{BCRHPbpm}
b_+(k)=\om_+(k)-\id,\quad b_-(k)=\id-\om_-(k).
\ee
And define
\be\label{BCRHPom}
\om(k)=\om_+(k)+\om_-(k).
\ee
\par
Let
\be\label{BCRHPcauchy}
(C_{\pm}f)(k)=\int_{\Gam}\frac{f(\xi)}{\xi-k_{\pm}}\frac{d\xi}{2\pi i},\quad k\in\Gam,f\in L^2(\Gam),
\ee
denote the Cauchy operator on $\Gam$. As is well known, the operator $C_{\pm}$ are bounded from $L^2(\Gam)$ to $L^2(\Gam)$, and $C_+-C_-=\Rmnum{1}$, here $\Rmnum{1}$ denote the identify operator.
\par
Define
\be\label{BCRHPCom}
C_{\om}f=C_+(f\om_-)+C_-(f\om_+)
\ee
for $2\times 2$ matrix-valued functions $f$. Let $\mu$ be the solution of the basic inverse equation
\be\label{BCRHPinverseeq}
\mu=\id+C_{\om}\mu.
\ee
Then
\be\label{BCRHPsol}
M(k)=\id+\int_{\Gam}\frac{\mu(\xi)\om(\xi)}{\xi-k}\frac{d\xi}{2\pi i},\quad k\in \bar \C\backslash \Gam,
\ee
is the solution of the Riemann-Hilbert problem. (See \cite{dz},P.322).

\subsection{The Riemann-Hilbert problem for short pulse equation}
Let us define
\be\label{Mdef}
M(x,t,k)=\left\{\ba{cc}\left(\ba{cc}[\mu_2]_1&\frac{[\mu_1]_2}{a(k)}\ea\right),&k\in D_1,\\
\left(\ba{cc}\frac{[\mu_2]_1}{\ol{a(\bar k)}}&[\mu_1]_2\ea\right),&k\in D_2.
\ea
\right.
\ee
From the definition (\ref{Mdef}) and (\ref{mujdef}), we can deduce $M(x,t,k)$ satisfies the symmetry condition
\be\label{msymcon}
  \ol{M(x,t,\bar k)}=M(x,t,-k)=\sig_2M(x,t,k)\sig_2.
  \ee

And $M(x,t,k)$ satisfies the following Riemann-Hilbert problem:
\begin{itemize}
  \item Jump condition: The two limiting values
                        \be
                        M_{\pm}(x,t,k)=\lim_{\eps\rightarrow 0}M_{\pm}(x,t,k\pm i\eps),\quad k\in \R,
                        \ee
                        are related by
\be\label{Mjump}
M_+(x,t,k)=M_-(x,t,k)J(x,t,k),\quad k\in \R,
\ee
where 
\be\label{Jdef}
J(x,t,k)=e^{ikp(x,t,k)\hat \sig_3}J_0(k)
\ee
here
\be
J_0(k)=\left(\ba{cc}
1&r(k)\\
\ol{r(k)}&1+|r(k)|^2
\ea
\right)
\ee
with $r(k)=\frac{b(k)}{a(k)}$.
  \item Normalize condition as $k\rightarrow \infty$
\be
M(x,t,k)=\id+O(\frac{1}{k}).
\ee
\end{itemize}
In order to get the information of the solution $u(x,t)$, we should consider the asymptotic behavior of $M(x,t,k)$ as $k\rightarrow 0$, that is,
{\small
  \be\label{Masyk0}
  M(x,t,k)=G^{-1}(x,t)\left[\id+k(ic_+\sig_3+iu\sig_1)+k^2[-\frac{c^2_++u^2}{2}\id+i(uc_+-2u_t+u^2u_x)\sig_2]+O(k^3)\right],
  \ee
}
  where
  \be
  c_+=\int_x^{+\infty}(\sqrt{m(x',t)}-1)dx'.
  \ee

Equations (\ref{Masyk0}) show that the matrix-valued function $M(x,t,k)$ contains all necessary information for reconstructing the solution of the initial value problem of (\ref{spe})-(\ref{spe-ini}) in terms of the solution of a matrix-valued Riemann-Hilbert problem.

However, the jump relation (\ref{Jdef}) cannot
be used immediately for recovering the solution of SP equation (\ref{spe})-(\ref{spe-ini}). Since, in the representation of the jump matrix
$e^{ikp(x,t,k)\hat \sig_3}J_0(k)$ the factor $J_0(k)$ is indeed given in terms of the known initial
data $u_0(x)$ but $p(x,t,k)$ is not, it involves $m(x,t)$ which is unknown (and, in fact, is
to be reconstructed).
\par
To overcome this, we introduce the new (time-dependent) scale
\be\label{ydef}
y(x,t)=x-\int_{x}^{+\infty}(\sqrt{m(x',t)}-1)dx'=x-c_+(x,t).
\ee
in terms of which the jump matrix becomes explicit.
The price to pay for this, however, is that the solution of the initial problem
can be given only implicitly, or parametrically: it will be given in terms of
functions in the new scale, whereas the original scale will also be given in terms
of functions in the new scale.

By the definition of the new scale $y(x,t)$, we define
\be\label{muydef}
\tilde M(y,t,k)=M(x(y,t),t,k),
\ee
then we can obtain the Riemann-Hilbert problem of $\tilde M(y,t,k)$ as follows:
  \begin{itemize}
  \item Analyticity: $\tilde M(y,t,k)$ is analytic in the two open half-planes $D_1$ and $D_2$, and continuous up to the boundary $k\in\R$.
  \item Jump condition: The two limiting values
                        \begin{subequations}\label{muyrhp}
                        \be
                        \tilde M_{+}(y,t,k)=\tilde M_{-}(y,t,k)\tilde J(y,t,k),\quad k\in\R,
                        \ee
                        where the jump matrix is
                        \be
                        \tilde J(y,t,k)=e^{i(ky-\frac{t}{4k})\hat\sig_3}J_0(k)
                        \ee
                        with
                        \be\label{J0def}
                        J_0(k)=\left(\ba{cc}
                                     1&r(k)\\
                                     \ol{r(k)}&1+|r(k)|^2
                                     \ea
                                     \right)
                        \ee
                        \end{subequations}
  \item Normalization:
                       \be
                       \tilde M(y,t,k)\rightarrow \id,\quad k\rightarrow \infty.
                       \ee
  \end{itemize}

\begin{theorem}
Let $\tilde M(y,t,k)$ satisfies the above conditions, then this Riemann-Hilbert problem has a unique solution. And the solution $u(x,t)$ of the initial value problem (\ref{spe})-(\ref{spe-ini}) can be expressed, in parametric form, in terms of the solution of this Rieamnn-Hilbert problem:
\begin{subequations}
\be
u(x,t)=u(y(x,t),t),
\ee
where
\be
x(y,t)=y+\lim_{k\rightarrow 0}\frac{\left((\tilde M(y,t,0))^{-1}\tilde M(y,t,k)\right)_{11}}{ik}
\ee
\be
u(y,t)=\lim_{k\rightarrow 0}\frac{\left((\tilde M(y,t,0))^{-1}\tilde M(y,t,k)\right)_{21}}{ik}
\ee
\end{subequations}

\end{theorem}

\begin{proof}
 Since the jump matrix $\tilde J(y,t,k)$ is a Hermitian matrix, then the Riemann-Hilbert problem of $\tilde M(y,t,k)$ indeed has a solution. Furthermore, the Riemann-Hilbert problem has only one solution because of the normalize condition.

 \par
 The statements of the solution $u(x,t)$ is following from the asymptotic formula (\ref{Masyk0}).

\end{proof}

\section{Long-time Asymptotics: fast decaying region $\xi=\frac{x}{t}>\eps>0$, Proof of theorem \ref{main1}}
In this section, we employ the nonlinear steepest descent method introduced by Deift and Zhou \cite{dz} to analyze the long-time asymptotic behavior of the solution $u(x,t)$ to the initial value problem (\ref{spe})-(\ref{spe-ini}).
\par
The key feature of the method is the deformation of the original Riemann-Hilbert problem according to the signature table for the phase function $\theta$ in jump matrix $\tilde J$ written in the form
\be
\tilde J(y,t,k)=e^{it\tha(\tilde\xi,k)\hat\sig_3}J_0(k),
\ee
where
\be
\tha(\tilde\xi,k)=\tilde\xi k-\frac{1}{4k},
\ee
\be
\tilde\xi=\frac{y}{t}.
\ee
\par
The signature table is the distribution of signs of $\im \tha(\tilde\xi,k)$ in the $k-$plane, 
\be
\im \tha(\tilde\xi,k)=k_2[\tilde\xi+\frac{1}{4(k_1^2+k_2^2)}],
\ee
where $k_1$ and $k_2$ are the real and image part of $k$, respectively, i.e. $k=k_1+ik_2$.
\par

Under the condition $\tilde \xi>\eps$ for any $\eps>0$, then we have $\im \tha(\tilde\xi,k)>0$ and $\im \tha(\tilde\xi,k)<0$, as $k_2=\im k>0$ and $k_2=\im k<0$, respectively, see figure \ref{fig3}.

\begin{figure}[th]
\centering
\includegraphics{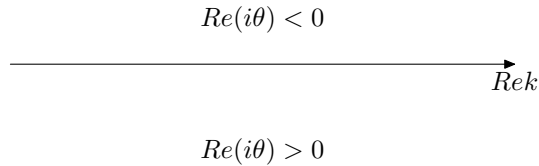}
\caption{\small The signs of $\re i\tha$ in the $k-$plane in the case $\tilde\xi >0$.}\label{fig3}
\end{figure}

This suggests the use of the following factorization of the jump matrix for all $k\in \R$:
\be\label{Jycase1fac}
\tilde J(y,t,k)=
\left(\ba{cc}1&0\\\ol{r(k)}e^{-2it\tha}&1\ea\right)\left(\ba{cc}1&r(k)e^{2it\tha}\\0&1\ea\right)
\ee
\par

\begin{figure}[th]
\centering
\includegraphics{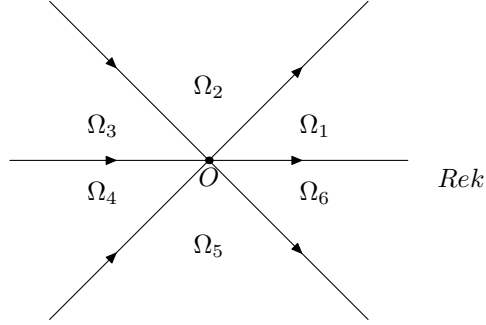}
\caption{\small The contour $\tilde \Sigma$ in the $k-$plane as $\tilde\xi >0$.}\label{fig4}
\end{figure}
Without loss of generality, we may assume that the function $\ol{r(k)}e^{-2it\tha}$ of the left triangular factor of (\ref{Jycase1fac}) extends analytically to the region $\im k<0$ and continuous in the closure of the region. And the function $r(k)e^{2it\tha}$ of the right factor extends the region $\im k>0$ by taking conjugate. If the analytic conditions are dropped, then there is a similar procedure to obtain the "weak" analytic conditions. To show this, we write $\ol{r(k)}$ as a Fourier transform with respect to $\tha$,
\be
\ba{rcl}
\ol{r(k)}e^{-2it\tha}&=&\frac{e^{-2it\tha}}{\sqrt{2\pi}(k-i)^2}\int_{-\infty}^{\infty}e^{is\tha(k)}\hat{g}(s)ds\\
{}&=&\frac{e^{-2it\tha}}{\sqrt{2\pi}(k-i)^2}\int_{t}^{\infty}e^{is\tha(k)}\hat{g}(s)ds+\frac{e^{-2it\tha}}{\sqrt{2\pi}(k-i)^2}\int_{-\infty}^{t}e^{is\tha(k)}\hat{g}(s)ds\\
{}&=&e^{-2it\tha(k)}h_{\Rmnum{1}}(k)+e^{-2it\tha(k)}h_{\Rmnum{2}}(k),
\ea
\ee
where
\[
\ba{l}
\hat{g}(s)=\frac{1}{\sqrt{2\pi}}\int_{-\infty}^{\infty}e^{-is\tha(k)}g(\tha)d\tha,\\
g(\tha)=\ol{r(k(\tha))}(k(\tha)-i)^2.
\ea
\]

Here $e^{-2it\tha(k)}h_{\Rmnum{2}}(k)$ has an analytic continuation to the lower half-plane and decays exponentially in $L^{1}\cap L^{\infty}(\Sig\cap\{k|\im k<0\})$, as $t\rightarrow \infty$, while $e^{-2it\tha(k)}h_{\Rmnum{1}}(k)$ decays rapidly in $L^{1}\cap L^{\infty}(\R)$, as $t\rightarrow \infty$.

Introducing the following the following transformation:
\be
\tilde M^{(1)}(y,t,k)=\left\{
\ba{ll}
\tilde M(y,t,k) \left(\ba{cc}1&-\ol{h_{\Rmnum{2}}(\bar k)}e^{2it\tha}\\0&1\ea\right),&k\in \Omega_1\cup \Omega_3,\\
\tilde M(y,t,k) \left(\ba{cc}1&0\\h_{\Rmnum{2}}(k)e^{-2it\tha}&1\ea\right),& k\in \Omega_4\cup \Omega_6,\\
\tilde M(y,t,k),& k\in \Omega_2\cup \Omega_5,
\ea
\right.
\ee
where $\Omega_j,j=1,2,\dots,6$ are shown in figure \ref{fig4}.
We obtain the new Riemann-Hilbert problem for $\tilde M^{(1)}(y,t,k)$,
\be
\left\{
\ba{l}
\tilde M_+^{(1)}(y,t,k)=\tilde M_-^{(1)}(y,t,k) \tilde J^{(1)}(y,t,k),\quad k\in \tilde\Sigma,\\
\tilde M^{(1)}(y,t,k)\rightarrow \id,\quad k\rightarrow \infty.
\ea
\right.
\ee
where
\be
\tilde J^{(1)}(y,t,k)=\left\{
\ba{ll}
\left(\ba{cc}1&\ol{h_{\Rmnum{2}}(\bar k)}e^{2it\tha}\\0&1\ea\right),&k\in \tilde\Sigma\cap D_1,\\
\left(\ba{cc}1&0\\h_{\Rmnum{2}}(k)e^{-2it\tha}&1\ea\right),&k\in \tilde\Sigma\cap D_2,\\
\left(\ba{cc}1&0\\h_{\Rmnum{1}}(k)e^{-2it\tha}&1\ea\right)\left(\ba{cc}1&\ol{h_{\Rmnum{1}}(\bar k)}e^{2it\tha}\\0&1\ea\right),&k\in \R,
\ea
\right.
\ee
\begin{theorem}
As $t\rightarrow \infty$, the solution $u(x,t)$ of the initial value problem (\ref{spe})-(\ref{spe-ini}) decays fast in the range $\xi>\eps$ for any $\eps >0$.
\end{theorem}
\begin{proof}
The above transformation reduces the Riemann-Hilbert problem of $\tilde M^{(1)}(y,t,k)$ to that with exponentially decaying in t to the identity matrix jump matrix. Since this Riemann-Hilbert problem is holomorphic, its solution decays fast to $\id$ and consequently $\tilde u(y,t)$ decays fast to $0$ while $y$ approaches fast $x$ and thus the domain $\tilde \xi >\eps$ and $\xi>\eps$ coincide asymptotically.
\end{proof}

\section{Long-time Asymptotics: oscillation region $\xi<-\eps<0$, Proof of theorem \ref{main2}}

If $\tilde \xi<-\eps$ for any $\eps>0$, let $k_0$ be defined by
         \be\label{k0def}
         k_0=\sqrt{\frac{-1}{4\tilde\xi}},
         \ee
then the signature table is shown as figure \ref{fig1},
\begin{figure}[th]
\centering
\includegraphics{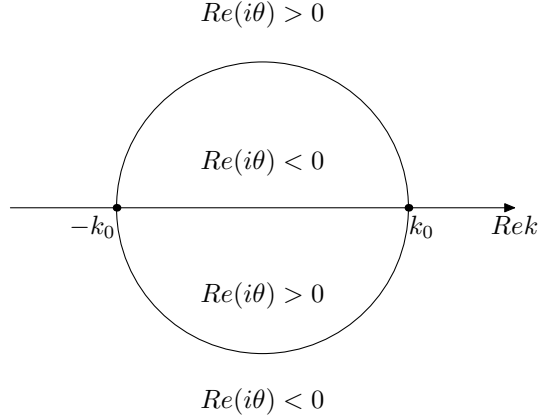}
\caption{\small The signs of $\re i\tha$ in the $k-$plane in the case $\tilde\xi <0$.}\label{fig1}
\end{figure}
\par
This suggests the use of the following factorizations of the jump matrix $\tilde J(y,t,k)$:
{\small
\be\label{Jycase2fac}
\tilde J(y,t,k)=\left\{
\ba{ll}
\left(\ba{cc}1&0\\\ol{r(k)}e^{-2it\tha}&1\ea\right)\left(\ba{cc}1&r(k)e^{2it\tha}\\0&1\ea\right),&|k|<k_0,\\
\left(\ba{cc}1&\frac{r(k)}{1+|r(k)|^2}e^{2it\tha}\\0&1\ea\right)\left(\ba{cc}\frac{1}{1+|r(k)|^2}&0\\0&1+|r(k)|^2\ea\right)
\left(\ba{cc}1&0\\\frac{\ol{r(k)}}{1+|r(k)|^2}e^{-2it\tha}&1\ea\right),&|k|>k_0.
\ea
\right.
\ee
}
Then we need make some appropriate sequence of deformations of the Riemann-Hilbert problem. 

\subsection{The conjugate transformation}
The aim of the first transformation involves the removal of the diagonal factor in (\ref{Jycase2fac}) for $|k|>k_0$.

Introducing a scalar function $\delta(k)$ which satisfies the following scalar Riemann-Hilbert problem
\be
\left\{
\ba{rll}
\dta_+(k)&=\dta_-(k)(1+|r(k)|^2),&|k|>k_0,\\
&=\dta_-(k)=\dta(k),&|k|<k_0.\\
\dta(k)&\rightarrow 1&k\rightarrow\infty.
\ea
\right.
\ee
Then the function $\dta(k)$ is given by
\be\label{dtadef}
\dta(k)=e^{\frac{1}{2\pi i}(\int_{-\infty}^{-k_0}+\int_{k_0}^{+\infty})\frac{\ln{(1+|r(s)|^2)}}{s-k}ds}.
\ee
The conjugate transformation
\be\label{m1trans}
\tilde M^{(1)}(y,t,k)=\tilde M(y,t,k)\delta(k)^{\sig_3},
\ee
yields the Riemann-Hilbert for $\tilde M^{(1)}(y,t,k)$
\begin{subequations}\label{mu1rhp}
\be\label{Jy1rhp}
\left\{
\ba{l}
\tilde M^{(1)}_+(y,t,k)=\tilde M^{(1)}_-(y,t,k)\tilde J^{(1)}(y,t,k),\quad k\in\R,\\
\tilde M^{(1)}(y,t,k)\rightarrow \id,\quad k\rightarrow \infty,
\ea
\right.
\ee
where
\be
\tilde J^{(1)}(y,t,k)=\left\{
\ba{ll}
\left(\ba{cc}1&0\\\ol{r(k)}\dta^2e^{-2it\tha}&1\ea\right)\left(\ba{cc}1&r(k)\dta^{-2}e^{2it\tha}\\0&1\ea\right),&|k|<k_0,\\
\left(\ba{cc}1&\frac{r(k)}{1+|r(k)|^2}\dta_-^{-2}e^{2it\tha}\\0&1\ea\right)
\left(\ba{cc}1&0\\\frac{\ol{r(k)}}{1+|r(k)|^2}\dta_+^2e^{-2it\tha}&1\ea\right),&|k|>k_0.
\ea
\right.
\ee
\end{subequations}

Now, let us come back to the solution $u(x,t)$.
From (\ref{dtadef}) it follows that
\be
\dta(k)=\dta_0+k\dta_1+O(k^2)=1-\frac{ik}{\pi}\int_{k_0}^{\infty}\frac{\ln{(1+|r(s)|^2)}}{s^2}ds+O(k^2),
\ee
If we write
\be
\tilde M(y,t,k)=\tilde M_0(y,t)+k\tilde M_1(y,t)+O(k^2),\quad k\rightarrow 0,
\ee
and
\be
\tilde M^{(1)}(y,t,k)=\tilde M^{(1)}_0(y,t)+k\tilde M^{(1)}_1(y,t)+O(k^2),\quad k\rightarrow 0,
\ee
then from the transformation (\ref{m1trans}) we obtain
\be
\tilde M_0(y,t)=\tilde M^{(1)}_0(y,t),\quad \tilde M_1(y,t)=\tilde M^{(1)}_1(y,t)-\tilde M^{(1)}_0(y,t)\dta_1\sig_3,
\ee
Hence, we have
\begin{subequations}
\be
u(x,t)=-i\left[(\tilde M^{(1)}_0)^{-1}\tilde M^{(1)}_1\right]_{21},
\ee
\be
c_+=-i\left(\left[(\tilde M^{(1)}_0)^{-1}\tilde M^{(1)}_1\right]_{11}-\dta_1\right)
\ee
\end{subequations}

\subsection{Analytic extension}
For the convenience of the notation, we transverse the direction of the component $|k|<k_0$ of the jump contour $\R$ for the Riemann-Hilbert problem of $\tilde M^{(1)}(y,t,k)$. Then, the jump matrix $\tilde J^{(1)}(y,t,k)$ becomes
\be\label{Jy1jump}
\tilde J^{(1)}(y,t,k)=\left\{
\ba{ll}
\left(\ba{cc}1&-r(k)\dta^{-2}e^{2it\tha}\\0&1\ea\right)\left(\ba{cc}1&0\\-\ol{r(k)}\dta^2e^{-2it\tha}&1\ea\right),&|k|<k_0,\\
\left(\ba{cc}1&\frac{r(k)}{1+|r(k)|^2}\dta_-^{-2}e^{2it\tha}\\0&1\ea\right)
\left(\ba{cc}1&0\\\frac{\ol{r(k)}}{1+|r(k)|^2}\dta_+^2e^{-2it\tha}&1\ea\right),&|k|>k_0.
\ea
\right.
\ee

Denote contours:
\begin{subequations}
\be
L_{0}=\{k|k=k_0\lam e^{-i\frac{\pi}{4}},0\le k\le \frac{1}{\sqrt{2}}\}\cup \{k|k=k_0\lam e^{i\frac{3\pi}{4}},0\le k\le \frac{1}{\sqrt{2}}\},
\ee
\be
L_{1}=\{k|k=k_0+k_0\lam e^{-i\frac{3\pi}{4}},-\infty<\lam<\frac{1}{\sqrt{2}}\},\quad L_{1\eps}=\{k|k=k_0+k_0\lam e^{-i\frac{3\pi}{4}},\eps<\lam<\frac{1}{\sqrt{2}}\}
\ee
\be
L_{2}=\{k|k=-k_0+k_0\lam e^{-i\frac{\pi}{4}},-\infty<\lam<\frac{1}{\sqrt{2}}\},\quad L_{2\eps}=\{k|k=-k_0+k_0\lam e^{-i\frac{\pi}{4}},\eps<\lam<\frac{1}{\sqrt{2}}\}.
\ee
\end{subequations}

\begin{proposition}\label{analyticpro}
Let
\be\label{rhodef}
\rho(k)=\left\{
\ba{ll}
-\ol{r(k)},&|k|<k_0,\\
\frac{\ol{r(k)}}{1+|r(k)|^2},&|k|>k_0.
\ea
\right.
\ee
Then $\rho(k)$ has a decomposition
\be\label{rhoanalic}
\rho(k)=h_{\Rmnum{1}}(k)+(h_{\Rmnum{2}}(k)+R(k)),
\ee
where $h_{\Rmnum{1}}(k)$ is small and $h_{\Rmnum{2}}(k)$ has an analytic continuation to $L$ and $L_0$. For example, if $k>k_0$, $h_{\Rmnum{2}}(k)$ of the function $\rho(k)$ has an analytic continuation to the $L_1\cap\im k>0$. And $R(k)$ is piecewise rational ($R(k)=0$, if $k\in L_0$) function.
\par
And let $M$ is a positive constant, as $k_0<M$, $R(k),h_{\Rmnum{1}}(k), h_{\Rmnum{2}}(k)$ satisfy
\begin{subequations}
\be\label{case2h1}
|e^{-2it\tha(k)}h_{\Rmnum{1}}(k)|\le\frac{c}{(1+|k|^2)t^l}, \quad k\in \R,
\ee

\be
|e^{-2it\tha(k)}\frac{h_{\Rmnum{1}}(k)}{k^j}|\le\frac{c}{(1+|k|^2)t^l},\quad 0<|k|<\frac{|k_0|}{2},j=1,2,
\ee

\be\label{case2h2l}
|e^{-2it\tha(k)}h_{\Rmnum{2}}(k)|\le \frac{c}{(1+|k|^2)t^l},\quad k\in L=\L_1\cup L_2,
\ee
\be\label{case2h2l0}
|e^{-2it\tha(k)}h_{\Rmnum{2}}(k)| \le ce^{-t\frac{1}{8k_0}},\quad k\in L_0,
\ee

\be
|e^{-2it\tha(k)}\frac{h_{\Rmnum{2}}(k)}{k^j}| \le ce^{-t\frac{1}{8k_0}},\quad k\in L_0,j=1,2,
\ee

and
\be\label{case2R}
|e^{-2it\tha(k)}R(k)|\le ce^{-\frac{\eps^2}{2k_0}t},\quad k\in L_{\eps}=L_{1\eps}\cup L_{2\eps}.
\ee
\end{subequations}
for arbitrary natural number $l$, for sufficiently large constants $c$, for some fixed positive constant $M$.
\end{proposition}

\begin{proof}
See the appendix \ref{appda}.
\end{proof}

\subsection{Second Transformation}
The main purpose of this subsection is to reformulate the Riemann-Hilbert problem for $\tilde M^{(1)}(y,t,k)$ (\ref{mu1rhp}) as an equivalent Riemann-
Hilbert problem on the augmented contour $\Sigma=L\cup\bar L\cup L_0\cup \bar L_0\cup \R$ (see figure \ref{fig2}).

According to the above analytic extension of $\rho(k)$, we can write the jump matrix $\tilde J^{(1)}(y,t,k)$ as
\begin{subequations}
\be
\tilde J^{(1)}(y,t,k)=b_-^{-1}(y,t,k)b_+(y,t,k),
\ee
where
\be
b_-(y,t,k)=\left(\ba{cc}1&-\ol{h_{\Rmnum{1}}(\bar k)}\dta_-^{-2}e^{2it\tha}\\0&1\ea\right)\left(\ba{cc}1&-\ol{h_{\Rmnum{2}}+R}(\bar k)\dta_-^{-2}e^{2it\tha}\\0&1\ea\right)=b_-^{R}(y,t,k)b_-^{I}(y,t,k),
\ee
and
\be
b_-(y,t,k)=\left(\ba{cc}1&0\\h_{\Rmnum{1}}(k)\dta_+^{2}e^{-2it\tha}&1\ea\right)\left(\ba{cc}1&0\\(h_{\Rmnum{2}}+R)(k)\dta_+^{2}e^{-2it\tha}&1\ea\right)=b_+^{R}(y,t,k)b_+^{I}(y,t,k),
\ee
\end{subequations}
\par
We make a transformation as
\be
\tilde M^{(2)}(y,t,k)=\tilde M^{(1)}(y,t,k)T(y,t,k),
\ee
where
\be
T(y,t,k)=\left\{
\ba{ll}
(b_+^{I}(y,t,k))^{-1},&k\in \Omega_1\cup \Omega_{3}\cup\Omega_9\cup \Omega_{10},\\
(b_-^{I}(y,t,k))^{-1},&k\in \Omega_4\cup \Omega_6\cup\Omega_7\cup \Omega_8,\\
\id,&k\in \Omega_2\cup \Omega_5.
\ea
\right.
\ee
with the regions $\{\Omega_j\}_{j=1}^{10}$ defined as figure \ref{fig2}.
\begin{figure}[th]
\centering
\includegraphics{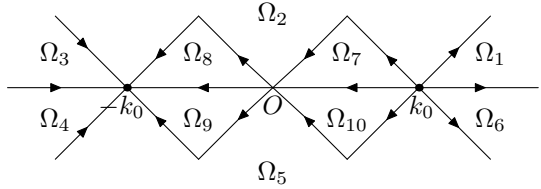}
\caption{\small The jump contour $\Sigma$ for $\tilde M^{(2)}(y,t,k)$.}\label{fig2}
\end{figure}

Then, $\tilde M^{(2)}(y,t,k)$ satisfies the following Riemann-Hilbert problem,
\begin{subequations}\label{m2rhp}
\be
\left\{
\ba{l}
\tilde M^{(2)}_+(y,t,k)=\tilde M^{(2)}_-(y,t,k)\tilde J^{(2)}(y,t,k),\\
\tilde M^{(2)}(y,t,k)\rightarrow \id,\quad k\rightarrow \infty.
\ea
\right.
\ee
where
\be
\tilde J^{(2)}(y,t,k)=(b^{(2)}_-(y,t,k))^{-1}b^{(2)}_+(y,t,k)=\left\{
\ba{ll}
(b_-^{R}(y,t,k))^{-1}b_+^{R}(y,t,k),&k\in \R,\\
b_+^{I}(y,t,k),&k\in \Sigma\cap\im k>0,\\
(b_-^{I}(y,t,k))^{-1},&k\in \Sigma\cap\im k<0.
\ea
\right.
\ee
\end{subequations}

\begin{proposition}
The reflection coefficient $r(k)=O(k^3)$ as $k\rightarrow 0$.
\end{proposition}
\begin{proof}
A direct calculation following from (\ref{akasy}) and from the identity $|r(k)|^2=\frac{1}{|a(k)|^2}-1$.
\end{proof}

Now, let us come back to the considered problem in this paper again. The solution $u(y,t)$ is related to the solution of the Riemann-Hilbert problem evaluated at $k=0$, it may be affected by this transformation. However, due to the above fact, the second transformation turns out not to affect the terms in the expansion of the solution of the Riemann-Hilbert problem at $k=0$ at least up to the terms of order $O(k^2)$ and thus it does not really affect $u(y,t)$.

So, if we write
\be
\tilde M^{(2)}(y,t,k)=\tilde M^{(2)}_0(y,t)+k\tilde M^{(2)}_1(y,t)+O(k^2),\quad k\rightarrow 0,
\ee
then we have
\begin{subequations}
\be
u(x,t)=-i\left[(\tilde M^{(2)}_0)^{-1}\tilde M^{(2)}_1\right]_{21},
\ee
\be
c_+=-i\left(\left[(\tilde M^{(2)}_0)^{-1}\tilde M^{(2)}_1\right]_{11}-\dta_1\right)
\ee
\end{subequations}

Set
\be
\omega_{\pm}^{(2)}(y,t,k)=\pm(b^{(2)}_{\pm}(y,t,k)-\id),\quad \omega=\omega^{(2)}_++\omega^{(2)}_-,
\ee
and let $\mu^{(2)}(y,t,k)$ be the solution of the singular integral equation $\mu^{(2)}=\id+C_{\omega}\mu^{(2)}$, here $C_{\omega}$ is defined as (\ref{BCRHPCom}),
then
\be\label{m2sol}
\tilde M^{(2)}(y,t,k)=\id+\frac{1}{2\pi i}\int_{\Sigma}\frac{\mu^{(2)}(y,t,\eta)\omega(y,t,\eta)}{\eta-k}d\eta,\quad k\in \C\backslash \Sigma
\ee
is the solution of Riemann-Hilbert problem (\ref{m2rhp}).
\par
Expanding the integral (\ref{m2sol}) around $k=0$, we have
\begin{subequations}\label{m20and1}
\be
\tilde M^{(2)}_0(y,t)=\id+\frac{1}{2\pi i}\int_{\Sigma}\frac{\mu^{(2)}(y,t,\eta)\omega(y,t,\eta)}{\eta}d\eta,
\ee
\be
\tilde M^{(2)}_1(y,t)=\frac{1}{2\pi i}\int_{\Sigma}\frac{\mu^{(2)}(y,t,\eta)\omega(y,t,\eta)}{\eta^2}d\eta
\ee
\end{subequations}
\begin{remark}
Since, $\omega(y,t,k)$ decays rapidly at $0$, the integral (\ref{m20and1}) are nonsingular.
\end{remark}

\subsection{Reduction to the cross}
\begin{figure}[th]
\centering
\includegraphics{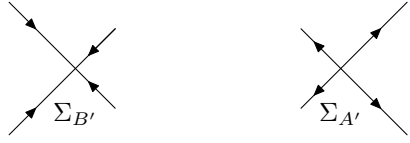}
\caption{\small The jump contour $\Sigma^{(3)}$ for $\tilde M^{(3)}(y,t,k)$.}\label{fig2}
\end{figure}
Let $\om^e$ be a sum of four terms
\be
\om^e=\om^a+\om^b+\om^c+\om^d.
\ee
We then have the following:
\be
\ba{l}
\om^a=\om \mbox{ is supported on the $\R$ and consists of terms of type $h_{\Rmnum{1}}(k)$ and $\ol{ h_{\Rmnum{1}}(k)}$}.\\
\om^b=\om \mbox{ is supported on the $L\cup \bar L$ and consists of terms of type $h_{\Rmnum{2}}(k)$ and $\ol{ h_{\Rmnum{2}}(\bar k)}$}.\\
\om^c=\om \mbox{ is supported on the $L_{\eps}\cup \bar L_{\eps}$ and consists of terms of type $R(k)$ and $\ol{R(\bar k)}$}.\\
\om^d=\om \mbox{ is supported on the $L_0\cup \bar L_0$ }.
\ea
\ee
Set $\om'=\om-\om^e$. Then, $\om'=0$ on $\Sig^{(2)}\backslash \Sig^{(3)}$.
Thus, $\om'$ is supported on $\Sig^{(3)}$ with contribution to $\om$ from rational
terms $R$ and $\bar R$.
\begin{proposition}\label{omfanshu}
For $0<k_0<M$, we have
\begin{subequations}
\be\label{oma}
||\om^a||_{L^1\cap L^2\cap L^{\infty}(\R)}\le \frac{c}{t^l},
\ee
\be
||\frac{\om^a}{k^j}||_{L^1\cap L^2\cap L^{\infty}(|k|<k_0)}\le \frac{c}{t^l},j=1,2
\ee

\be\label{omb}
||\om^b||_{L^1(L\cup \bar L)\cap L^2(L\cup \bar L)\cap L^{\infty}(L\cup \bar L)}\le \frac{c}{t^l},
\ee
\be\label{omc}
||\om^c||_{L^1(L_{\eps}\cup \bar L_{\eps})\cap L^2(L_{\eps}\cup \bar L_{\eps})\cap L^{\infty}(L_{\eps}\cup \bar L_{\eps})}\le ce^{-\frac{\eps^2}{2k_0}t},
\ee
\be\label{omd}
||\om^d||_{L^1(L_{0}\cup \bar L_{0})\cap L^2(L_{0}\cup \bar L_{0})\cap L^{\infty}(L_{0}\cup \bar L_{0})}\le ce^{-\frac{t}{8k_0}},
\ee
\be
||\frac{\om^d}{k^j}||_{L^1(L_{0}\cup \bar L_{0})\cap L^2(L_{0}\cup \bar L_{0})\cap L^{\infty}(L_{0}\cup \bar L_{0})}\le ce^{-\frac{t}{8k_0}},j=1,2
\ee

\end{subequations}
Moreover,
\be\label{om'}
||\om'||_{L^2(\Sig^{(3)})}\le \frac{c}{t^{\frac{1}{4}}},\qquad ||\om'||_{L^1(\Sig^{(3)})}\le \frac{c}{t^{\frac{1}{2}}}
\ee
\end{proposition}
\begin{proof}
Consequence of proposition \ref{analyticpro}, and analogous calculations as in lemma 2.13 of \cite{dz}. 
\end{proof}
\begin{proposition}
As $t\rightarrow \infty$ and $0<k_0<M$, $||(1-C_{\om'})^{-1}||_{L^2(\Sig^{(2)})}$ exists and is uniformly bounded, and $||(1-C_{\om})^{-1}||_{L^2(\Sig^{(2)})}\le C$ is equivalent to $||(1-C_{\om'})^{-1}||_{L^2(\Sig^{(2)})}\le C$.
\end{proposition}
\begin{proof}
The existence of the operator $(1-C_{\om'})^{-1}$ is followed similar to \cite{dz} P.324. And the equivalence is the consequence of the following inequality, $||C_{\om}-C_{\om'}||_{L^2(\Sig^{(2)})}\le c||\om^e||_{L^2(\Sig^{(2)})}$,
the fact that $||\om^e||_{L^2(\Sig^{(2)})}\le \frac{c}{t^l}$, and the second resolvent identity.
\end{proof}
\begin{proposition}
If $||(1-C_{\om'})^{-1}||_{L^2(\Sig^{(2)})}\le C$, then for arbitrary positive integer $l$,
as $t\rightarrow \infty$ such that $0<k_0<M$,
\be \label{m3sig2om'}
\int_{\Sig}\frac{((\id-C_{\om})^{-1}\id)(\eta)\om(x,t,\eta)}{\eta^j}d\eta=\int_{\Sig}\frac{((\id-C_{\om^{'}})^{-1}\id)(\eta)\om^{'}(x,t,\eta)}{\eta^j}d\eta+O(\frac{c}{t^l}),j=1,2.
\ee
\end{proposition}
\begin{proof}
From the second resolvent identity, one can derive the following expression (see equation (2.27) in \cite{dz}),
\be\label{m2}
\ba{rrl}
\int_{\Sig}\frac{((1-C_{\om})^{-1}\id)\om}{\eta^j}d\eta&=&\int_{\Sig}\frac{((1-C_{\om'})^{-1}\id)\om'}{\eta^j}d\eta+\int_{\Sig}\frac{\om^e}{\eta^j}d\eta\\
&&+\int_{\Sig}\frac{((1-C_{\om'})^{-1}(C_{\om^e}\id))\om}{\eta^j}d\eta\\
&&+\int_{\Sig}\frac{((1-C_{\om'})^{-1}(C_{\om'}\id))\om^e}{\eta^j}d\eta\\
&&+\int_{\Sig}\frac{((1-C_{\om'})^{-1}C_{\om^e}(1-C_{\om})^{-1})(C_{\om}\id)\om}{\eta^j}d\eta\\
&=&\int_{\Sig}\frac{((1-C_{\om'})^{-1}\id)\om'}{\eta^j}d\eta+\Rmnum{1}+\Rmnum{2}+\Rmnum{3}+\Rmnum{4}.
\ea
\ee
For $0<k_0<M$, from Proposition (\ref{omfanshu}) it follows that,
\be\label{rmnum1}
\ba{rrl}
|\Rmnum{1}|&\le& ||\frac{\om^a}{k^j}||_{L^1(\R)}+||\frac{\om^b}{k^j}||_{L^1(L\cup \bar L)}+||\frac{\om^c}{k^j}||_{L^1(L_{\eps}\cup \bar L_{\eps})}+||\frac{\om^d}{k^j}||_{L^1(L_0\cup \bar L_0)}\\
&\le & ct^{-l},
\ea
\ee
\be\label{rmnum2}
\ba{rrl}
|\Rmnum{2}|&\le& ||(1-C_{\om'})^{-1}||_{L^2(\Sig)} ||(C_{\om^e}\id)||_{L^2(\Sig)} ||\frac{\om}{k^j}||_{L^2(\Sig)}\\
&\le & c||\om^e||_{L^2(\Sig^{(2)})} (||\om^e||_{L^2(\Sig)}+||\om'||_{L^2(\Sig)})\\
&\le & ct^{-l}(ct^{-l}+c)\le ct^{-l},
\ea
\ee
\be\label{rmnum3}
\ba{rrl}
|\Rmnum{3}|&\le & ||(1-C_{\om'})^{-1}||_{L^2(\Sig)} ||(C_{\om'}\id)||_{L^2(\Sig)} ||\frac{\om^e}{k^j}||_{L^2(\Sig)}\\
&\le & ct^{-l}
\ea
\ee
\be\label{rmnum4}
\ba{lll}
|\Rmnum{4}|&\le & ||(1-C_{\om'})^{-1}C_{\om^e}(1-C_{\om})^{-1})(C_{\om}\id)||_{L^2(\Sig)} ||\frac{\om}{k^j}||_{L^2(\Sig)}\\
&\le & ||(1-C_{\om'})^{-1}||_{L^2(\Sig)} ||C_{\om^e}||_{L^2(\Sig)} ||(1-C_{\om})^{-1}||_{L^2(\Sig)} ||(C_{\om}\id)||_{L^2(\Sig)} ||\frac{\om}{k^j}||_{L^2(\Sig)}\\
&\le & c ||C_{\om^e}||_{L^2(\Sig^{(2)})} ||(C_{\om}\id)||_{L^2(\Sig)} ||\frac{\om}{k^j}||_{L^2(\Sig^{(2)})}\\
&\le & c ||\om^e||_{L^2(\Sig)} ||\frac{\om}{k^j}||^2_{L^2(\Sig)}\\
&\le & c t^{-l}.
\ea
\ee
Hence,
\be
|\Rmnum{1}+\Rmnum{2}+\Rmnum{3}+\Rmnum{4}|\le ct^{-l}.
\ee
\par
Applying these estimates to equation (\ref{m2}), we can obtain equation (\ref{m3sig2om'}).
\end{proof}
\par

\par
Following the method in appendix \ref{appedb}, we apply the lemma \ref{opralema} to the case $u=\om'$, $\Sig_{12}=\Sig$ and $\Sig_{1}=\Sig^{(3)}$. From identity (\ref{ideopra2}),
we get the following proposition, which shows that the integral region can be changed from $\Sig$ to $\Sig^{(3)}$ without alternating the Riemann-Hilbert problem.
\begin{proposition}\label{m3pro}
\be\label{m3}
\int_{\Sig}\frac{((\id-C_{\om^{'}})^{-1}\id)(\eta)\om^{'}(x,t,\eta)}{\eta^j}d\eta=\int_{\Sig^{(3)}}\frac{((\id-C_{\om^{'}})^{-1}\id)(\eta)\om^{'}(x,t,\eta)}{\eta^j}d\eta.
\ee
\end{proposition}
\par
Set
$$L'=L\backslash L_{\eps}$$.
Then, $\Sig^{(3)}=L'\cup \bar L'$. On $\Sig^{(3)}$, set
$\mu^{'}=(1^{\Sig^{(3)}}-C^{\Sig^{(3)}}_{\om'})^{-1}\id$. Then,
\be
M^{(3)}(x,t,k)=\id+\int_{\Sig^{(3)}}\frac{\mu^{'}(\xi)\om'(\xi)}{\xi-k}\frac{d\xi}{2\pi i}
\ee
solves the Riemann-Hilbert problem
\be\label{M3RHP}
\left\{
\ba{ll}
M^{(3)}_+(x,t,k)=M^{(3)}_-(x,t,k)J^{(3)}(x,t,k),& k\in\Sig^{(3)},\\
M^{(3)}\rightarrow \id,&k\rightarrow \infty.
\ea
\right.
\ee
where
\bea\label{M3canshu}
&\om'=\om'_++\om'_-,&\\
&b'_{\pm}=\id \pm \om'_{\pm},&\\
&J^{(3)}(x,t,k)=(b'_-)^{-1}b'_+&
\eea

Hence, we have the representation of the solution is as follows,
\begin{theorem}
As $t\rightarrow \infty$,
\begin{subequations}\label{usolrep1}
\be\label{um3sol}
\ba{rcl}
iu(y,t)&=&\left(1+\frac{1}{2\pi i}\int_{\Sigma^{(3)}}\frac{(\id-C_{\om'})^{-1}\id(\eta)\om'(\eta)}{\eta}d\eta+O(t^{-l})\right)_{11}\\
{}&&\cdot \left(\frac{1}{2\pi i}\int_{\Sigma^{(3)}}\frac{(\id-C_{\om'})^{-1}\id(\eta)\om'(\eta)}{\eta^2}d\eta+O(t^{-l})\right)_{21}\\
{}&&-\left(\frac{1}{2\pi i}\int_{\Sigma^{(3)}}\frac{(\id-C_{\om'})^{-1}\id(\eta)\om'(\eta)}{\eta}d\eta+O(t^{-l})\right)_{21}\\
{}&&\cdot \left(\frac{1}{2\pi i}\int_{\Sigma^{(3)}}\frac{(\id-C_{\om'})^{-1}\id(\eta)\om'(\eta)}{\eta^2}d\eta+O(t^{-l})\right)_{11},
\ea
\ee
and
\be
\ba{rcl}
ic_+(y,t)&=&\left(1+\frac{1}{2\pi i}\int_{\Sigma^{(3)}}\frac{(\id-C_{\om'})^{-1}\id(\eta)\om'(\eta)}{\eta}d\eta+O(t^{-l})\right)_{22}\\
{}&&\cdot \left(\frac{1}{2\pi i}\int_{\Sigma^{(3)}}\frac{(\id-C_{\om'})^{-1}\id(\eta)\om'(\eta)}{\eta^2}d\eta+O(t^{-l})\right)_{11}\\
{}&&-\left(\frac{1}{2\pi i}\int_{\Sigma^{(3)}}\frac{(\id-C_{\om'})^{-1}\id(\eta)\om'(\eta)}{\eta}d\eta+O(t^{-l})\right)_{12}\\
{}&&\cdot \left(\frac{1}{2\pi i}\int_{\Sigma^{(3)}}\frac{(\id-C_{\om'})^{-1}\id(\eta)\om'(\eta)}{\eta^2}d\eta+O(t^{-l})\right)_{21}-\dta_1.
\ea
\ee
\end{subequations}
\end{theorem}

\subsection{Separate out the contributions of the two crosses}
Using the estimates of the proposition \ref{omfanshu} and the similar method in \cite{dz} P.330-331, we can separate out the contributions of the two crosses in $\Sig^{(3)}$ to the solution $u(y,t)$ in formula (\ref{um3sol}). Let the contour $\Sig^{(3)}=\Sig^{A'}\cup \Sig^{B'}$ and write
\be
\om'=\om^{A'}+\om^{B'},
\ee
where
\be
\ba{ll}
\om^{A'}(k)=0,&\mbox{for }k\in \Sig_{B'},\\
\om^{B'}(k)=0,&\mbox{for }k\in \Sig_{A'}.
\ea
\ee
\begin{proposition}
\be
\ba{l}
||C^{\Sig^{(3)}}_{\omega^{B'}}C^{\Sig^{(3)}}_{\omega^{A'}}||_{L^2(\Sig^{(3)})}=||C^{\Sig^{(3)}}_{\omega^{A'}}C^{\Sig^{(3)}}_{\omega^{B'}}||_{L^2(\Sig^{(3)})}\le \frac{C(k_0)}{\sqrt{t}},\\
||C^{\Sig^{(3)}}_{\omega^{B'}}C^{\Sig^{(3)}}_{\omega^{A'}}||_{L^{\infty}\rightarrow L^2(\Sig^{(3)})},\quad ||C^{\Sig^{(3)}}_{\omega^{A'}}C^{\Sig^{(3)}}_{\omega^{B'}}||_{L^{\infty}\rightarrow L^2(\Sig^{(3)})}\le \frac{C(k_0)}{t^{3/4}}.
\ea
\ee
\end{proposition}

\begin{proof}
Since
\begin{subequations}
\be
\omega^{B'}_+(\eta)\omega^{A'}_+(\xi)=0,\quad \omega^{B'}_-(\eta)\omega^{A'}_-(\xi)=0,\quad \mbox{for }\eta,\xi\in \Sig^{(3)},
\ee
and
\be
C^{\Sig^{(3)}}_{\omega^{A'}}C^{\Sig^{(3)}}_{\omega^{B'}}\phi=C_+\left((C_-\phi\omega_+^{B'})\omega_-^{A'}\right)+C_-\left((C_+\phi\omega_-^{B'})\omega_+^{A'}\right)
\ee
\end{subequations}
Here we estimate the first term and the second term is similar.

\par
Since $C_-$ is bounded in $L^{2}(\Sig^{(3)})$, and the proposition \ref{omfanshu}, we have
\be
\ba{l}
||C_+\left((C_-\phi\omega_+^{B'})\omega_-^{A'}\right)(k)||_{L^{2}(\Sig^{(3)})}\\
=||\int_{\Sig^{A'}}\left(\int_{\Sig^{B'}}\phi(\xi)\omega_+^{B'}(\xi)\frac{d\xi}{(\xi-\eta)_-}\right)\omega_-^{A'}(\eta)\frac{d\eta}{(\eta-k)_+}||_{L^{2}(\Sig^{(3)})}\\
\le c ||\omega_-^{A'}||_{L^{2}(\Sig^{A'})}\sup_{\eta\in\Sig^{A'}}{|\int_{\Sig^{B'}}\phi(\xi)\omega_+^{B'}(\xi)\frac{d\xi}{\xi-\eta}|}\\
\le \frac{c}{k_0}||\omega_-^{A'}||_{L^{2}(\Sig^{A'})} ||\omega_+^{B'}||_{L^{2}(\Sig^{B'})}||\phi||_{L^2(\Sig^{(3)})}\\
\le C(k_0)t^{-1/2}||\phi||_{L^2(\Sig^{(3)})},
\ea
\ee
where
\be
\mbox{dist}(\Sig^{A'},\Sig^{B'})>k_0.
\ee
Thus, we have
\be
||C^{\Sig^{(3)}}_{\omega^{A'}}C^{\Sig^{(3)}}_{\omega^{B'}}||_{L^2(\Sig^{(3)})}\le \frac{C(k_0)}{\sqrt{t}}.
\ee

On the other hand.
\be
\ba{l}
||C_+\left((C_-\phi\omega_+^{B'})\omega_-^{A'}\right)(k)||_{L^{2}(\Sig^{(3)})}\\
=||\int_{\Sig^{A'}}\left(\int_{\Sig^{B'}}\phi(\xi)\omega_+^{B'}(\xi)\frac{d\xi}{(\xi-\eta)_-}\right)\omega_-^{A'}(\eta)\frac{d\eta}{(\eta-k)_+}||_{L^{2}(\Sig^{(3)})}\\
\le c ||\omega_-^{A'}||_{L^{2}(\Sig^{A'})}\sup_{\eta\in\Sig^{A'}}{|\int_{\Sig^{B'}}\phi(\xi)\omega_+^{B'}(\xi)\frac{d\xi}{\xi-\eta}|}\\
\le \frac{c}{k_0}||\omega_-^{A'}||_{L^{2}(\Sig^{A'})} ||\omega_-^{B'}||_{L^{1}(\Sig^{B'})}||\phi||_{L^{\infty}(\Sig^{(3)})}\\
\le C(k_0)t^{-1/2}t^{-1/4}    ||\phi||_{L^{\infty}(\Sig^{(3)})},
\ea
\ee
Thus, we have
\be
||C^{\Sig^{(3)}}_{\omega^{A'}}C^{\Sig^{(3)}}_{\omega^{B'}}||_{L^{\infty}\rightarrow L^2(\Sig^{(3)})}\le \frac{C(k_0)}{t^{3/4}}.
\ee
\end{proof}

Using the identity
\be
\ba{l}
\left(\id-C^{\Sig^{(3)}}_{\omega^{A'}}-C^{\Sig^{(3)}}_{\omega^{B'}}\right)\left(\id+C^{\Sig^{(3)}}_{\omega^{A'}}(\id-C^{\Sig^{(3)}}_{\omega^{A'}})^{-1}+C^{\Sig^{(3)}}_{\omega^{B'}}(\id-C^{\Sig^{(3)}}_{\omega^{B'}})^{-1}\right)\\
=\id-C^{\Sig^{(3)}}_{\omega^{B'}}C^{\Sig^{(3)}}_{\omega^{A'}}(\id-C^{\Sig^{(3)}}_{\omega^{A'}})^{-1}-C^{\Sig^{(3)}}_{\omega^{A'}}C^{\Sig^{(3)}}_{\omega^{B'}}(\id-C^{\Sig^{(3)}}_{\omega^{B'}})^{-1}
\ea
\ee
and proposition \ref{omfanshu}, we show that as $t\rightarrow \infty$,
\be
\ba{l}
\frac{1}{2\pi i}\int_{\Sigma^{(3)}}\frac{(\id-C_{\om^{'}})^{-1}\id(\eta)\om^{'}(\eta)}{\eta^{j}}d\eta\\
=\frac{1}{2\pi i}\int_{\Sigma^{A'}}\frac{(\id-C_{\om^{A'}})^{-1}\id(\eta)\om^{A'}(\eta)}{\eta^{j}}d\eta\\
+\frac{1}{2\pi i}\int_{\Sigma^{A'}}\frac{(\id-C_{\om^{A'}})^{-1}\id(\eta)\om^{A'}(\eta)}{\eta^{j}}d\eta+O(\frac{C(k_0)}{t}),\quad j=1,2,
\ea
\ee
where $C(k_0)$ is a constant dependent on $k_0$.

\par
Then, using the lemma \ref{opralema},  we obtain
\begin{proposition}
As $t\rightarrow \infty$,
\begin{subequations}
\be\label{um3'sol}
\ba{rcl}
iu(y,t)&=&\left(1+\frac{1}{2\pi i}\int_{\Sigma^{A'}}\frac{(\id-C_{\om^{A'}})^{-1}\id(\eta)\om^{A'}(\eta)}{\eta}d\eta+\frac{1}{2\pi i}\int_{\Sigma^{B'}}\frac{(\id-C_{\om^{B'}})^{-1}\id(\eta)\om^{B'}(\eta)}{\eta}d\eta+O(t^{-l})\right)_{11}\\
{}&&\cdot \left(\frac{1}{2\pi i}\int_{\Sigma^{A'}}\frac{(\id-C_{\om^{A'}})^{-1}\id(\eta)\om^{A'}(\eta)}{\eta^2}d\eta+\frac{1}{2\pi i}\int_{\Sigma^{B'}}\frac{(\id-C_{\om^{B'}})^{-1}\id(\eta)\om^{B'}(\eta)}{\eta^2}d\eta+O(t^{-l})\right)_{21}\\
{}&&-\left(\frac{1}{2\pi i}\int_{\Sigma^{A'}}\frac{(\id-C_{\om^{A'}})^{-1}\id(\eta)\om^{A'}(\eta)}{\eta}d\eta+\frac{1}{2\pi i}\int_{\Sigma^{B'}}\frac{(\id-C_{\om^{B'}})^{-1}\id(\eta)\om^{B'}(\eta)}{\eta}d\eta+O(t^{-l})\right)_{21}\\
{}&&\cdot \left(\frac{1}{2\pi i}\int_{\Sigma^{A'}}\frac{(\id-C_{\om^{A'}})^{-1}\id(\eta)\om^{A'}(\eta)}{\eta^2}d\eta+\frac{1}{2\pi i}\int_{\Sigma^{B'}}\frac{(\id-C_{\om^{B'}})^{-1}\id(\eta)\om^{B'}(\eta)}{\eta^2}d\eta+O(t^{-l})\right)_{11},
\ea
\ee
and
\be
\ba{rcl}
ic_+(y,t)&=&\left(1+\frac{1}{2\pi i}\int_{\Sigma^{A'}}\frac{(\id-C_{\om^{A'}})^{-1}\id(\eta)\om^{A'}(\eta)}{\eta}d\eta+\frac{1}{2\pi i}\int_{\Sigma^{B'}}\frac{(\id-C_{\om^{B'}})^{-1}\id(\eta)\om^{B'}(\eta)}{\eta}d\eta+O(t^{-l})\right)_{22}\\
{}&&\cdot \left(\frac{1}{2\pi i}\int_{\Sigma^{A'}}\frac{(\id-C_{\om^{A'}})^{-1}\id(\eta)\om^{A'}(\eta)}{\eta^2}d\eta+\frac{1}{2\pi i}\int_{\Sigma^{B'}}\frac{(\id-C_{\om^{B'}})^{-1}\id(\eta)\om^{B'}(\eta)}{\eta^2}d\eta+O(t^{-l})\right)_{11}\\
{}&&-\left(\frac{1}{2\pi i}\int_{\Sigma^{A'}}\frac{(\id-C_{\om^{A'}})^{-1}\id(\eta)\om^{A'}(\eta)}{\eta}d\eta+\frac{1}{2\pi i}\int_{\Sigma^{B'}}\frac{(\id-C_{\om^{B'}})^{-1}\id(\eta)\om^{B'}(\eta)}{\eta}d\eta+O(t^{-l})\right)_{12}\\
{}&&\cdot \left(\frac{1}{2\pi i}\int_{\Sigma^{A'}}\frac{(\id-C_{\om^{A'}})^{-1}\id(\eta)\om^{A'}(\eta)}{\eta^2}d\eta+\frac{1}{2\pi i}\int_{\Sigma^{B'}}\frac{(\id-C_{\om^{B'}})^{-1}\id(\eta)\om^{B'}(\eta)}{\eta^2}d\eta+O(t^{-l})\right)_{21}\\
{}&&-\dta_1.
\ea
\ee
\end{subequations}
\end{proposition}

\subsection{The scaling transformation}
In order to reduce the Riemann-Hilbert problem for $\tilde M^{(3)}(y,t,k)$, as $t\rightarrow \infty$, to a model Riemann-Hilbert problem whose solution can be given explicitly in terms of parabolic cylinder functions, see \cite{dz}, the leading term of the factor $\dta(k)e^{-it\tha(k)}$ as $k\rightarrow \pm k_0$ is to be evaluated.

First, we extend the crosses $\Sig^{A'}$ and $\Sig^{B'}$ to contours $\hat \Sig^{A'}$ and $\hat \Sig^{B'}$ by zero extension. Thus, the corresponding functions $\hat \om^{A'}$ and $\hat \om^{B'}$ are well-defined by zero extension of the functions $\om^{A'}$ and $\om^{B'}$, too. Then, we denote $\Sig^{A}$ and $\Sig^{B}$ as the contour $\{k=k_0\lam e^{\pm\frac{\pi i}{4}},-\infty<\lam<\infty\}$ oriented as $\hat \Sig^{A'}$ and $\hat \Sig^{B'}$, respectively.

\par
For $k$ near $k_0$,
\be
\dta(k)=\left(\frac{k-k_0}{k+k_0}\right)^{-i\nu(k_0)}e^{\chi(k)},
\ee
where
\be\label{nudef}
\nu(k_0)=\nu=-\frac{1}{2\pi}\ln{(1+|r(k_0)|^2)},
\ee
\be\label{chidef}
\chi(k)=-\frac{1}{2\pi i}(\int_{-\infty}^{-k_0}+\int_{k_0}^{+\infty})\ln{|k-s|}d\ln{(1+|r(s)|^2)}.
\ee
And
\be
\tha(k)=-\frac{1}{2k_0}-\frac{1}{4k^3_0}(k-k_0)^2+\frac{1}{4\eta^4}(k-k_0)^3,\quad \eta \mbox{ lies between $k_0$ and $k$}.
\ee
Then, introducing the scaling operator by
\be
(N_Af)(k)=f(k_0+\frac{k}{\sqrt{k_0^{-3}t}}),
\ee
the factor $\dta(k)e^{-it\tha(k)}$ can be scaled as
\be
(N_A\dta e^{-it\tha})(k)=\dta^0_A\dta^1_A,
\ee
where
\begin{subequations}
\be\label{dtaA0def}
\dta^0_A=(\frac{4t}{k_0})^{\frac{i\nu(k_0)}{2}}e^{\chi(k_0)}e^{\frac{it}{2k_0}},
\ee
\be\label{dtaA1def}
\dta^1_A=k^{-i\nu(k_0)}e^{\frac{ik^2}{4}}\left(\frac{2k_0}{2k_0+\frac{k}{\sqrt{k_0^3t}}}\right)^{-i\nu(k_0)}e^{\chi(k_0+\frac{k}{\sqrt{k_0^{-3}t}})-\chi(k_0)}e^{-\frac{ik^3}{4\eta^4k_0^{-9/2}t^{1/2}}}.
\ee
\end{subequations}
Here $k^{-i\nu(k_0)}$ is cut along $(0,\infty)$.
\par
Form the definition of $\chi(k)$, we know that $\chi(k_0)$ is purely imaginary, thus $|\dta^0_A|=1$. Define
\be
\Dta^0_A=(\dta^0_A)^{-\sig_3},\quad \tilde \Dta^0_A \phi=\phi \Dta^0_A
\ee
We have
\be
C_{\hat \om^{A'}}=N^{-1}_A(\Dta^0_A)^{-1}A\tilde \Dta^0_A N_A,
\ee
where the operator $A:L^{2}(\Sig^A)\rightarrow L^2(\Sig^A)$ is given by
\be
A\phi=C_{(\Dta^0_A)^{-1}(N_A\hat\om^{A'})\Dta^0_A}\phi=C_+(\phi (\Dta^0_A)^{-1}(N_A\hat\om_-^{A'})\Dta^0_A )+C_-(\phi (\Dta^0_A)^{-1}(N_A\hat\om_+^{A'})\Dta^0_A )
\ee
On the part $\{k=k_0\lam e^{\frac{i\pi}{4}},-\eps<\lam<\eps\}$ of $\Sig^A$,
\be
(\Dta^0_A)^{-1}(N_A\hat\om_+^{A'})\Dta^0_A=\left(\ba{cc}0&0\\R(k_0+\frac{k}{\sqrt{k_0^{-3}t}})(\dta^1_A)^{2}&1\ea\right),
\ee
otherwise, $(\Dta^0_A)^{-1}(N_A\hat\om_+^{A'})\Dta^0_A=0$.
\par
Similarly, on the part $\{k=k_0\lam e^{-\frac{i\pi}{4}},-\eps<\lam<\eps\}$ of $\Sig^A$,
\be
(\Dta^0_A)^{-1}(N_A\hat\om_-^{A'})\Dta^0_A=\left(\ba{cc}0&-\ol{R}(k_0+\frac{k}{\sqrt{k_0^{-3}t}})(\dta^1_A)^{-2}\\0&1\ea\right),
\ee
otherwise, $(\Dta^0_A)^{-1}(N_A\hat\om_-^{A'})\Dta^0_A=0$.
\par
By the definition of $R(k)$, we have
\be
R(k_0+)=\lim_{\re k>k_0}R(k)=\frac{\ol{r(k_0)}}{1+|r(k_0)|^2},
\ee
and
\be
R(k_0-)=\lim_{\re k<k_0}R(k)=-\ol{r(k_0)}.
\ee
\begin{proposition}\label{tasyinftyrsol}
As $t\rightarrow \infty$, let $\beta$ be a fixed small number, $0<2\beta<1$, then for $k\in\{k=k_0\lam e^{\frac{i\pi}{4}},-\eps<\lam<\eps\}$,
\be\label{Rasyt}
||R(k_0+\frac{k}{\sqrt{k_0^{-3}t}})(\dta^1_A)^{2}-R(k_0\pm)k^{-2i\nu(k_0)}e^{\frac{ik^2}{2}}||_{L^1\cap L^{\infty}}\le C(k_0)|e^{\frac{i\beta^2k^2}{2}}|\left(\frac{\ln{(t)}}{\sqrt{t}}\right).
\ee
\end{proposition}
\begin{proof}
See the appendix \ref{appdd}.
\end{proof}

\par

\par
Now, let us calculate
\begin{subequations} \label{msol1}
\be
\ba{l}
\frac{1}{2\pi i}\int_{\Sig^{A'}}\frac{((\id_{A'}-C^{\Sig^{A'}}_{\om^{A'}})^{-1}\id)(\eta)\om^{A'}(\eta)}{\eta}d\eta\\
=\frac{1}{2\pi i}\int_{\hat \Sig^{A'}}\frac{((\id_{\hat A'}-C^{\hat \Sig^{A'}}_{\hat\om^{A'}})^{-1}\id)(\eta)\hat\om^{A'}(\eta)}{\eta}d\eta\\
=\frac{1}{2\pi i}\int_{\hat \Sig^{A'}}\frac{(N_A^{-1}(\tilde \Dta^0_A)^{-1}(\id_A-A)^{-1}\tilde \Dta^0_AN_A\id)(\eta)\hat\om^{A'}(\eta)}{\eta}d\eta\\
=\frac{1}{2\pi i}\int_{\hat \Sig^{A'}}\frac{(\id_A-A)^{-1}\Dta_A^0((\eta-k_0)\sqrt{k_0^{-3}t})(\Dta_A^0)^{-1}(N_A\hat\om^{A'})((\eta-k_0)\sqrt{k_0^{-3}t})}{\eta}d\eta\\
=\frac{1}{2\pi i}\frac{1}{\sqrt{k_0^{-3}t}}\int_{\Sig^{A}}\frac{(\id_A-A)^{-1}\Dta_A^0(\eta)(\Dta_A^0)^{-1}(N_A\hat\om^{A'})(\eta)}{\frac{\eta}{\sqrt{k_0^{-3}t}}+k_0}d\eta\\
=\frac{1}{2\pi i}\frac{1}{\sqrt{k_0^{-3}t}}\Dta_A^0\left(\int_{\Sig^{A}}\frac{(\id_A-A)^{-1}\id(\eta)\om^{A}(\eta)}{\frac{\eta}{\sqrt{k_0^{-3}t}}+k_0}d\eta\right)(\Dta_A^0)^{-1}
\ea
\ee
and
\be
\ba{l}
\frac{1}{2\pi i}\int_{\Sig^{A'}}\frac{((\id_{A'}-C^{\Sig^{A'}}_{\om^{A'}})^{-1}\id)(\eta)\om^{A'}(\eta)}{\eta^2}d\eta\\
=\frac{1}{2\pi i}\int_{\hat \Sig^{A'}}\frac{((\id_{\hat A'}-C^{\hat \Sig^{A'}}_{\hat\om^{A'}})^{-1}\id)(\eta)\hat\om^{A'}(\eta)}{\eta^2}d\eta\\
=\frac{1}{2\pi i}\int_{\hat \Sig^{A'}}\frac{(N_A^{-1}(\tilde \Dta^0_A)^{-1}(\id_A-A)^{-1}\tilde \Dta^0_AN_A\id)(\eta)\hat\om^{A'}(\eta)}{\eta^2}d\eta\\
=\frac{1}{2\pi i}\int_{\hat \Sig^{A'}}\frac{(\id_A-A)^{-1}\Dta_A^0((\eta-k_0)\sqrt{k_0^{-3}t})(\Dta_A^0)^{-1}(N_A\hat\om^{A'})((\eta-k_0)\sqrt{k_0^{-3}t})}{\eta^2}d\eta\\
=\frac{1}{2\pi i}\frac{1}{\sqrt{k_0^{-3}t}}\int_{\Sig^{A}}\frac{(\id_A-A)^{-1}\Dta_A^0(\eta)(\Dta_A^0)^{-1}(N_A\hat\om^{A'})(\eta)}{\left(\frac{\eta}{\sqrt{k_0^{-3}t}}+k_0\right)^2}d\eta\\
=\frac{1}{2\pi i}\frac{1}{\sqrt{k_0^{-3}t}}\Dta_A^0\left(\int_{\Sig^{A}}\frac{(\id_A-A)^{-1}\id(\eta)\om^{A}(\eta)}{\left(\frac{\eta}{\sqrt{k_0^{-3}t}}+k_0\right)^2}d\eta\right)(\Dta_A^0)^{-1}
\ea
\ee
\end{subequations}
By the proposition \ref{tasyinftyrsol}, we have
\begin{subequations}
\be
\int_{\Sig^{A}}\frac{(\id_A-A)^{-1}\id(\eta)\om^{A}(\eta)}{\frac{\eta}{\sqrt{k_0^{-3}t}}+k_0}d\eta=\int_{\Sig^{A}}\frac{(\id_A-A)^{-1}\id(\eta)\om^{A^0}(\eta)}{\frac{\eta}{\sqrt{k_0^{-3}t}}+k_0}d\eta+O\left(C(k_0)\frac{\ln{(t)}}{\sqrt{t}}\right),
\ee
\be
\int_{\Sig^{A}}\frac{(\id_A-A)^{-1}\id(\eta)\om^{A}(\eta)}{\left(\frac{\eta}{\sqrt{k_0^{-3}t}}+k_0\right)^2}d\eta=\int_{\Sig^{A}}\frac{(\id_A-A)^{-1}\id(\eta)\om^{A^0}(\eta)}{\left(\frac{\eta}{\sqrt{k_0^{-3}t}}+k_0\right)^2}d\eta+O\left(C(k_0)\frac{\ln{(t)}}{\sqrt{t}}\right).
\ee
\end{subequations}

Here
\be
\om^{A^0}=\left\{
\ba{ll}
\om^{A^0}_+=\left(\ba{cc}0&0\\\frac{\ol{r(k_0)}}{1+|r(k_0)|^2}k^{-2i\nu(k_0)}e^{\frac{ik^2}{2}}&0\ea\right),&k\in \Sig_A^{1},\\
\om^{A^0}_+=\left(\ba{cc}0&0\\-\ol{r(k_0)}k^{-2i\nu(k_0)}e^{\frac{ik^2}{2}}&0\ea\right),&k\in \Sig_A^{3},\\
\om^{A^0}_-=\left(\ba{cc}0&-r(k_0)k^{2i\nu(k_0)}e^{-\frac{ik^2}{2}}\\0&0\ea\right),&k\in \Sig_A^{2},\\
\om^{A^0}_+=\left(\ba{cc}0&\frac{r(k_0)}{1+|r(k_0)|^2}k^{2i\nu(k_0)}e^{-\frac{ik^2}{2}}\\0&0\ea\right),&k\in \Sig_A^{4},\\
\ea
\right.
\ee
\begin{figure}[th]
\centering
\includegraphics{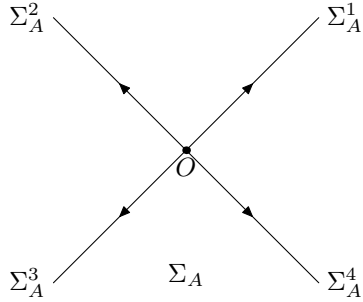}
\caption{\small The  $\Sigma^{A}$.}\label{fig6}
\end{figure}

Define
\be
\tilde M^{A_0}(k)=\id+\frac{1}{2\pi i}\int_{\Sig^{A}}\frac{((\id_A-A^0)^{-1}\id)(\eta)\om^{A^0}(\eta)}{\eta-k}d\eta,
\ee
then $\tilde M^{A_0}(k)$ satisfies the Riemann-Hilbert problem
\be\label{ma0rhp}
\left\{
\ba{ll}
\tilde M_+^{A^0}(k)=\tilde M_-^{A^0}(k)\tilde J^{A_0}(k),&k\in \Sig^{A},\\
\tilde M^{A^0}(k)\rightarrow \id,&k\rightarrow \infty,
\ea
\right.
\ee
where
\be
\tilde J^{A_0}(k)=(b^{A^0}_-)^{-1}(k)b^{A^0}_+(k)=(\id-\om_-^{A^0})^{-1}(\id+\om_+^{A^0}).
\ee
\par
If
\be
\tilde M^{A^0}(k)=\id-\frac{\tilde M^{A^0}_1}{k}+O(k^{-2}),\quad k\rightarrow \infty,
\ee
then
\begin{subequations}
\be
\ba{l}
\frac{1}{2\pi i}\frac{1}{\sqrt{k_0^{-3}t}}\int_{\Sig^{A}}\frac{(\id_A-A)^{-1}\id(\eta)\om^{A}(\eta)}{\frac{\eta}{\sqrt{k_0^{-3}t}}+k_0}d\eta\\
=\frac{1}{2\pi i}\int_{\Sig^{A}}\frac{(\id_A-A)^{-1}\id(\eta)\om^{A}(\eta)}{\eta+k_0{\sqrt{k_0^{-3}t}}}d\eta\\
=\tilde M^{A^0}(-k_0{\sqrt{k_0^{-3}t}})-\id\\
=\frac{1}{k_0{\sqrt{k_0^{-3}t}}}\tilde M^{A^0}_1+O(t^{-1}\ln{t}),\quad t\rightarrow \infty.
\ea
\ee
and
\be
\ba{l}
\frac{1}{2\pi i}\frac{1}{\sqrt{k_0^{-3}t}}\int_{\Sig^{A}}\frac{(\id_A-A)^{-1}\id(\eta)\om^{A}(\eta)}{\left(\frac{\eta}{\sqrt{k_0^{-3}t}}+k_0\right)^2}d\eta\\
=\frac{\sqrt{k_0^{-3}t}}{2\pi i}\int_{\Sig^{A}}\frac{(\id_A-A)^{-1}\id(\eta)\om^{A}(\eta)}{\left(\eta+k_0{\sqrt{k_0^{-3}t}}\right)^2}d\eta\\
=\left.\frac{d\tilde M^{A^0}}{dk}\right|_{k=-k_0{\sqrt{k_0^{-3}t}}}\\
=\frac{1}{k^2_0{\sqrt{k_0^{-3}t}}}\tilde M^{A^0}_1+O(t^{-1}\ln{t}),\quad t\rightarrow \infty.
\ea
\ee
\end{subequations}

\begin{remark}
Similarly for $k$ near $-k_0$. The scaling operator is
\be
(N_B f)(k)=f(-k_0+\frac{k}{\sqrt{k_0^{-3}t}}),
\ee
and
\be
(N_B\dta e^{-it\tha})(k)=\dta^0_B\dta^1_B,\rightarrow \tilde \dta\tilde k^{-i\nu(k_0)}e^{\frac{i\tilde k^2}{4}},\quad as\quad  t\rightarrow \infty,
\ee
where
\begin{subequations}
\be\label{dtaA0def}
\dta^0_B=(\frac{4t}{k_0})^{-\frac{i\nu(-k_0)}{2}}e^{\chi(-k_0)}e^{-\frac{it}{2k_0}},
\ee
\be\label{dtaA1def}
\dta^1_B=(-k)^{i\nu(k_0)}e^{-\frac{ik^2}{4}}\left(\frac{-2k_0}{-2k_0+\frac{k}{\sqrt{k_0^3t}}}\right)^{i\nu(k_0)}e^{\chi(-k_0+\frac{k}{\sqrt{k_0^{-3}t}})-\chi(-k_0)}e^{-\frac{ik^3}{4\eta^4k_0^{-9/2}t^{1/2}}}.
\ee
\end{subequations}

If
\be
\tilde M^{B^0}(k)=\id-\frac{\tilde M^{B^0}_1}{k}+O(k^{-2}),\quad k\rightarrow \infty,
\ee
then
\begin{subequations}\label{msol2}
\be
\ba{l}
\frac{1}{2\pi i}\frac{1}{\sqrt{k_0^{-3}t}}\int_{\Sig^{B}}\frac{(\id_B-B)^{-1}\id(\eta)\om^{B}(\eta)}{\frac{\eta}{\sqrt{k_0^{-3}t}}-k_0}d\eta\\
=\frac{1}{2\pi i}\int_{\Sig^{B}}\frac{(\id_B-B)^{-1}\id(\eta)\om^{A}(\eta)}{\eta-k_0{\sqrt{k_0^{-3}t}}}d\eta\\
=\tilde M^{B^0}(k_0{\sqrt{k_0^{-3}t}})-\id\\
=-\frac{1}{k_0{\sqrt{k_0^{-3}t}}}\tilde M^{B^0}_1+O(t^{-1}\ln{t}),\quad t\rightarrow \infty.
\ea
\ee
and
\be
\ba{l}
\frac{1}{2\pi i}\frac{1}{\sqrt{k_0^{-3}t}}\int_{\Sig^{B}}\frac{(\id_B-B)^{-1}\id(\eta)\om^{B}(\eta)}{\left(\frac{\eta}{\sqrt{k_0^{-3}t}}-k_0\right)^2}d\eta\\
=\frac{\sqrt{k_0^{-3}t}}{2\pi i}\int_{\Sig^{B}}\frac{(\id_B-B)^{-1}\id(\eta)\om^{B}(\eta)}{\left(\eta-k_0{\sqrt{k_0^{-3}t}}\right)^2}d\eta\\
=\left.\frac{d\tilde M^{B^0}}{dk}\right|_{k=k_0{\sqrt{k_0^{-3}t}}}\\
=\frac{1}{k^2_0{\sqrt{k_0^{-3}t}}}\tilde M^{B^0}_1+O(t^{-1}\ln{t}),\quad t\rightarrow \infty.
\ea
\ee
\end{subequations}
\end{remark}

\subsection{Model Riemann-Hilbert problem}
\par
Consider the Riemann-Hilbert problem (\ref{ma0rhp}) for $\tilde M^{A^0}$, we introduce a transformation
\be
\tilde \Phi^{A^0}(k)=\tilde M^{A^0}(k) \Phi^{T}(k),
\ee
where
\be
\Phi^{T}(k)=\left\{
\ba{ll}
k^{i\nu(k_0)\sig_3}e^{-\frac{ik^2}{4}\hat\sig_3}\left(\ba{cc}1&0\\\frac{\ol{r(k_0)}}{1+|r(k_0)|^2}&1\ea\right),&k\in \Omega^e_1,\\
k^{i\nu(k_0)\sig_3}e^{-\frac{ik^2}{4}\hat\sig_3}\left(\ba{cc}1&r(k_0)\\0&1\ea\right),&k\in \Omega^e_3,\\
k^{i\nu(k_0)\sig_3}e^{-\frac{ik^2}{4}\hat\sig_3}\left(\ba{cc}1&0\\-\ol{r(k_0)}&1\ea\right),&k\in \Omega^e_4,\\
k^{i\nu(k_0)\sig_3}e^{-\frac{ik^2}{4}\hat\sig_3}\left(\ba{cc}1&-\frac{r(k_0)}{1+|r(k_0)|^2}\\0&1\ea\right),&k\in \Omega^e_6,\\
k^{i\nu(k_0)\sig_3},&k\in \Omega^e_2\cup \Omega^e_5.
\ea
\right.
\ee

\begin{figure}[th]
\centering
\includegraphics{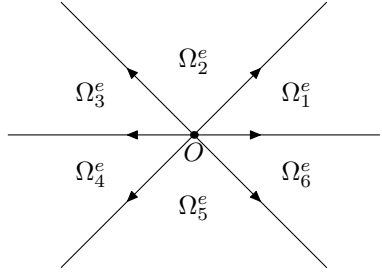}
\caption{\small The  regions $\Omega^{e}_j$, $j=1,2,\dots,6$.}\label{fig7}
\end{figure}

Then, $\tilde \Phi^{A^0}(k)$ satisfies the Riemann-Hilbert problem
\be
\left\{
\ba{ll}
\tilde \Phi_+^{A^0}(k)=\tilde \Phi_-^{A^0}(k)e^{-\frac{ik^2}{4}\hat\sig_3}\tilde J_{\Phi}^{A^0}(k),& k\in \R,\\
\tilde \Phi^{A^0}(k)\rightarrow k^{i\nu(k_0)\sig_3},&k\rightarrow \infty,
\ea
\right.
\ee
where
\be
\tilde J_{\Phi}^{A^0}(k)=\left\{
\ba{ll}
\left(\ba{cc}1&r(k_0)\\\ol{r(k_0)}&1+|r(k_0)|^2\ea\right),&k>0,\\
\left(\ba{cc}1+|r(k_0)|^2&-r(k_0)\\-\ol{r(k_0)}&1\ea\right),&k<0,
\ea
\right.
\ee
the contour $\R$ oriented from the original to $\infty$ and $-\infty$.
\par
Hence, if we reorient the contour $\R$ from $-\infty$ to $\infty$, we get the jump matrix
\be
\tilde J_{\Phi}^{A^0}(k)=\left(\ba{cc}1&r(k_0)\\\ol{r(k_0)}&1+|r(k_0)|^2\ea\right),\quad k\in \R.
\ee
\par
Let
\be
\tilde M^{model}(k)=\tilde \Phi^{A^0}(k)e^{-\frac{ik^2}{4}\sig_3},
\ee
then we have the model Riemann-Hilbert problem
\be
\tilde M_+^{model}(k)=\tilde M_-^{model}(k)\tilde J(k_0),\quad k\in \R,
\ee
where
\be
\tilde J(k_0)=\tilde J_{\Phi}^{A^0}(k)=\left(\ba{cc}1&r(k_0)\\\ol{r(k_0)}&1+|r(k_0)|^2\ea\right).
\ee

So, we have
\be
(\partial_{k}\tilde M^{model}+\frac{ik}{2}\sig_3\tilde M^{model})(\tilde M^{model})^{-1}=-\frac{i}{2}[\sig_3,\tilde M^{A^0}]+O(k^{-1}),
\ee
since $(\partial_{k}\tilde M^{model}+\frac{ik}{2}\sig_3\tilde M^{model})(\tilde M^{model})^{-1}$ is entire,
\be
\partial_{k}\tilde M^{model}+\frac{ik}{2}\sig_3\tilde M^{model}=\beta \tilde M^{model},
\ee
here
\be
\beta=-\frac{i}{2}[\sig_3,\tilde M^{A^0}]=\left(\ba{cc}0&\beta_{12}\\\beta_{21}&0\ea\right).
\ee
Thus, we have
\be
(\tilde M^{A^0})_{12}=i\beta_{12},\quad (\tilde M^{A^0})_{21}=-i\beta_{21}.
\ee

\par
Let us consider $\im k>0$, denote $\tilde M^{model}(k)$ by $\tilde M^{+}(k)$, we have
\begin{subequations}
\be
\partial_{k}\tilde M_{11}^{+}+\frac{ik}{2}\tilde M_{11}^{+}=\beta_{12} \tilde M_{21}^{+}
\ee
and
\be
\partial_{k}\tilde M_{21}^{+}-\frac{ik}{2}\tilde M_{21}^{+}=\beta_{21} \tilde M_{11}^{+},
\ee
\end{subequations}
so
\be
\partial^2_{k}\tilde M_{11}^{+}=(-\frac{k^2}{4}-\frac{i}{2}+\beta_{12} \beta_{21})\tilde M_{11}^{+}
\ee
Setting
\be
\tilde M_{11}^{+}=g(e^{-\frac{3i\pi}{4}}k),
\ee
we have the parabolic cylinder equation
\be
\partial^2_{\xi}+(\frac{1}{2}-\frac{\xi^2}{4}+a)g=0,
\ee
where $a=i\beta_{12}\beta_{21}$.
\par
Then,
\be
\tilde M_{11}^{+}(k)=c_1D_a(e^{-\frac{3i\pi}{4}}k)+c_2D_a(-e^{-\frac{3i\pi}{4}}k),
\ee
where $D_a(z)$ denotes the parabolic cylinder function.
\par
As $z\rightarrow \infty$, we have the asymptotic formula \cite{wwbook} P.327,
\be
\ba{rcl}
D_a(z)&=&z^{a}e^{-\frac{z^2}{4}}\left(1+O(z^{-2})\right),\quad |\mbox{arg}z|<\frac{3\pi}{4},\\
{}&=&z^ae^{-\frac{z^2}{4}}\left(1+O(z^{-2})\right)-\frac{\sqrt{2\pi}}{\Gamma(-a)}e^{a\pi i}z^{-a-1}e^{\frac{z^2}{4}}\left(1+O(z^{-2})\right),\quad \frac{\pi}{4}<\mbox{arg}z<\frac{5\pi}{4},\\
{}&=&z^ae^{-\frac{z^2}{4}}\left(1+O(z^{-2})\right)-\frac{\sqrt{2\pi}}{\Gamma(-a)}e^{-a\pi i}z^{-a-1}e^{\frac{z^2}{4}}\left(1+O(z^{-2})\right),\quad -\frac{5\pi}{4}<\mbox{arg}z<-\frac{\pi}{4}
\ea
\ee

We have
\be
a=i\nu(k_0),
\ee
so
\be
\nu(k_0)=\beta_{12}\beta_{21}.
\ee
Thus, for $\im k>0$,
\be
\ba{l}
\tilde M_{11}^{+}(k)=e^{-\frac{3\pi \nu(k_0)}{4}}D_a(e^{-\frac{3\pi i}{4}}k),\\
\tilde M_{21}^{+}(k)=\frac{1}{\beta_{12}}e^{-\frac{3\pi \nu(k_0)}{4}}(\partial_k D_a(e^{-\frac{3\pi i}{4}}k))+\frac{ik}{2}D_a(e^{-\frac{3\pi i}{4}}k),
\ea
\ee
Similarly, for $\im k<0$, we have
\be
\ba{l}
\tilde M_{11}^{-}(k)=e^{\frac{\pi \nu(k_0)}{4}}D_a(e^{\frac{\pi i}{4}}k),\\
\tilde M_{21}^{-}(k)=\frac{1}{\beta_{12}}e^{\frac{\pi \nu(k_0)}{4}}(\partial_k D_a(e^{\frac{\pi i}{4}}k))+\frac{ik}{2}D_a(e^{\frac{\pi i}{4}}k),
\ea
\ee

Since
\be
(\tilde M_-^{model})^{-1}\tilde M_+^{model}=\left(\ba{cc}1&r(k_0)\\\ol{r(k_0)}&1+|r(k_0)|^2\ea\right),
\ee
we have
\be
\ba{rcl}
\ol{r(k_0)}&=&\tilde M_{11}^{-}\tilde M_{21}^{+}-\tilde M_{21}^{-}\tilde M_{11}^{+}\\
{}&=&\frac{1}{\beta_{12}}e^{-\frac{\pi \nu(k_0)}{2}}\mbox{Wr}(D_a(e^{\frac{i\pi}{4}}k),D_a(e^{-\frac{3\pi i}{4}}k))\\
{}&=&\frac{\sqrt{2\pi}e^{\frac{i\pi}{4}}e^{-\frac{\pi \nu(k_0)}{2}}}{\beta_{12}\Gamma(-a)},
\ea
\ee
where $\mbox{Wr}(f,g)$ denotes the Wronskian of $f,g$ and $\Gam(\cdot)$ is the Euler Gamma function.

Hence, we have
\be
\beta_{12}=\frac{\sqrt{2\pi}e^{\frac{i\pi}{4}}e^{-\frac{\pi \nu(k_0)}{2}}}{\ol{r(k_0)}\Gamma(-i\nu(k_0))},
\ee
and
\be
\beta_{21}=\frac{\nu(k_0)}{\beta_{12}}=-\frac{\sqrt{2\pi}e^{-\frac{i\pi}{4}}e^{-\frac{\pi \nu(k_0)}{2}}}{r(k_0)\Gamma(i\nu(k_0))},
\ee
since $|\Gam(i\nu(k_0))|^2=\frac{\pi}{\nu(k_0)\mbox{sinh}(\pi\nu(k_0))}$.

\subsection{The asymptotic behavior of the solution $u(x,t)$}
Remind as $t\rightarrow \infty$, the representation (\ref{usolrep1}) of the solution $u(y,t)$ and the computation results (\ref{msol1}) and (\ref{msol2}), we have
\begin{subequations}
\be
\ba{rcl}
iu(y,t)&=&\left(1+\frac{1}{k_0\sqrt{k_0^{-3}t}}\tilde M_1^{A^0}-\frac{1}{k_0\sqrt{k_0^{-3}t}}\tilde M_1^{B^0}+O\left(C(k_0)\frac{\ln{t}}{t}\right)\right)_{11}\\
{}&&\cdot \left((\dta_A^0)^2\frac{1}{k^2_0\sqrt{k_0^{-3}t}}\tilde M_1^{A^0}+(\dta_B^0)^2\frac{1}{k^2_0\sqrt{k_0^{-3}t}}\tilde M_1^{B^0}+O\left(C(k_0)\frac{\ln{t}}{t}\right)\right)_{21}\\
{}&&-\left((\dta_A^0)^2\frac{1}{k_0\sqrt{k_0^{-3}t}}\tilde M_1^{A^0}-(\dta_B^0)^2\frac{1}{k_0\sqrt{k_0^{-3}t}}\tilde M_1^{B^0}+O\left(C(k_0)\frac{\ln{t}}{t}\right)\right)_{21}\\
{}&&\cdot \left(\frac{1}{k^2_0\sqrt{k_0^{-3}t}}\tilde M_1^{A^0}+\frac{1}{k^2_0\sqrt{k_0^{-3}t}}\tilde M_1^{B^0}+O\left(C(k_0)\frac{\ln{t}}{t}\right)\right)_{11},
\ea
\ee
and
\be
\ba{rcl}
ic_+(y,t)&=&\left(1+\frac{1}{k_0\sqrt{k_0^{-3}t}}\tilde M_1^{A^0}-\frac{1}{k_0\sqrt{k_0^{-3}t}}\tilde M_1^{B^0}+O\left(C(k_0)\frac{\ln{t}}{t}\right)\right)_{22}\\
{}&&\cdot \left(\frac{1}{k^2_0\sqrt{k_0^{-3}t}}\tilde M_1^{A^0}+\frac{1}{k^2_0\sqrt{k_0^{-3}t}}\tilde M_1^{B^0}+O\left(C(k_0)\frac{\ln{t}}{t}\right)\right)_{11}\\
{}&&-\left((\dta_A^0)^{-2}\frac{1}{k_0\sqrt{k_0^{-3}t}}\tilde M_1^{A^0}-(\dta_B^0)^{-2}\frac{1}{k_0\sqrt{k_0^{-3}t}}\tilde M_1^{B^0}+O\left(C(k_0)\frac{\ln{t}}{t}\right)\right)_{12}\\
{}&&\cdot \left((\dta_A^0)^2\frac{1}{k^2_0\sqrt{k_0^{-3}t}}\tilde M_1^{A^0}+(\dta_B^0)^2\frac{1}{k^2_0\sqrt{k_0^{-3}t}}\tilde M_1^{B^0}+O\left(C(k_0)\frac{\ln{t}}{t}\right)\right)_{21},\\
{}&&-\dta_1.
\ea
\ee
\end{subequations}

Notice that we have
\be
\dta^0_B=\ol{\dta^0_A},
\ee
as $\chi(-k_0)=-\chi(k_0)=\ol{\chi(k_0)}$.
\par
And from the symmetry conditions (\ref{msymcon}), we get
\be
\tilde M_1^{A^0}=-\ol{\tilde M_1^{B^0}}.
\ee

Hence, a direct computation shows that,
\be\label{uxtasy}
u(x,t)=\sqrt{\frac{-4\nu(\kappa_0)}{\kappa_0 t}}\sin{\{\frac{t}{\kappa_0}+\nu(\kappa_0)\ln{(\frac{4t}{\kappa_0})}+\phi(\kappa_0)\}}+O\left(\frac{\ln{(t)}}{t}\right),\quad \mbox{as }t\rightarrow \infty,
\ee
where
\be\label{phikappadef}
\phi(\kappa_0)=\frac{\pi}{4}-\arg{r(\kappa_0)}-\arg{\Gam(i\nu(\kappa_0))}+\frac{1}{\pi}(\int_{-\infty}^{\kappa_0}+\int_{\kappa_0}^{\infty})\ln{|\kappa_0-s|}d(1+|r(s)|^2)+2\kappa_0\Dta,
\ee
here
\be
\Dta=\frac{1}{\pi}\int_{\kappa_0}^{\infty}\frac{\ln{(1+|r(s)|^2)}}{s^2}ds, 
\ee
and $\kappa_0$ is defined as (\ref{kappa0def}).

This finishes the proof of theorem \ref{main2}.


\bigskip

{\bf Acknowledgements}
This work of Xu was supported by National
Science Foundation of China under project NO.11501365, Shanghai Sailing Program
supported by Science and Technology Commission of Shanghai Municipality
under Grant NO.15YF1408100, Shanghai youth teacher assistance program NO.ZZslg15056. And Xu want to give many thanks to Shanghai Center Mathematical Science, many of this work was done during Xu visited there.

\appendix
\section{Proof of the proposition \ref{analyticpro}}\label{appda}
\par
For the convenience of reader, we show the details of the procedure of the analytic continuation.
\par
\subsection{\bf 1. $\frac{k_0}{2}<|k|<k_0,k\in \R$.}
\par
We just consider $\frac{k_0}{2}<k<k_0$, the case for $-k_0<k<-\frac{k_0}{2}$ is similar.
\par
Set
\be\label{rholessk0}
\rho(k)=-\ol{r(k)}.
\ee
We split $\rho(k)$ into even and odd parts, $\rho(k)=H_{e}(k^2)+kH_{o}(k^2)$, where $H_{e}(\cdot)$ and $H_{o}(\cdot)$ are of the Schwartz class.
\par
For any positive integer $m$,
\be\label{He}
H_{e}(k^2)=\mu_0^{e}+\mu_1^{e}(k^2-k_0^2)+\cdots+\mu_m^{e}(k^2-k_0^2)^m+\frac{1}{m!}\int_{k_0^2}^{k^2}H_{e}^{(m+1)}(\gam)(k^2-\gam)^md\gam
\ee
and
\be\label{Ho}
H_{o}(k^2)=\mu_0^{o}+\mu_1^{o}(k^2-k_0^2)+\cdots+\mu_m^{o}(k^2-k_0^2)^m+\frac{1}{m!}\int_{k_0^2}^{k^2}H_{o}^{(m+1)}(\gam)(k^2-\gam)^md\gam.
\ee
Set
\be\label{Rm}
R(k)=R_m(k)=\sum_{i=0}^{m}\mu_i^{e}(k^2-k_0^2)^i+k\sum_{i=0}^{m}\mu_i^{o}(k^2-k_0^2)^i.
\ee
Assume $m=4q+1$, where $q$ is a positive integer. Write
\be
\rho(k)=h(k)+R(k),\quad \frac{k_0}{2}<k<k_0,k\in \R.
\ee
Then
\be
\left.\frac{d^j h(k)}{dk^j}\right|_{\pm k_0}=0,\quad 0\le j\le m.
\ee
And we have
\be
h(k)=\frac{(k^2-k_0^2)^{m+1}}{m!}g(k,k_0)
\ee
where
\small{
\be
g(k,k_0)=\left(\int_{0}^{1}H_{e}^{(m+1)}(k_0^2+u(k^2-k_0^2))(1-u)^{m}du
+k\int_{0}^{1}H_{o}^{(m+1)}(k_0^2+u(k^2-k_0^2))(1-u)^{m}du\right)
\ee
}
and
\be
\left|\frac{d^j g(k,k_0)}{dk^j}\right|\le C,\quad \frac{k_0}{2}\le k\le k_0.
\ee
We will split $h$ as $h(k)=h_{\Rmnum{1}}(k)+h_{\Rmnum{2}}(k)$, where $h_{\Rmnum{1}}$ is small and $h_{\Rmnum{2}}$ has an alytic continuation to $\im k<0$. Thus
\be
\rho=h_{\Rmnum{1}}+(h_{\Rmnum{2}}+R).
\ee
\par
Set $p(k)=(k^2-k_0^2)^q$. Recall
\be
\tha(k)=k\tilde \xi-\frac{1}{4k}=-\frac{1}{4k_0^2}(k+\frac{k_0^2}{k}).
\ee

\par
We define
\be
\left\{
\ba{rrll}
\frac{h}{p}(\tha)&=&\frac{h(k(\tha))}{p(k(\tha))},& \tha(\frac{k_0}{2})<\tha<\tha(k_0),\\
&=&0,& \mbox{otherwise}.
\ea
\right.
\ee
As $|\tha|\rightarrow \tha(k_0)=-\frac{1}{2k_0}$ and $|\tha|>\frac{1}{2k_0}$, we have $\frac{h}{p}(\tha)=O((k^2(\tha)-k_0^2)^{m+1-q})$ and
\be
\frac{d\tha}{dk}=\frac{k_0^2-k^2}{4k^2k_0^2}.
\ee
We claim that $\frac{h}{p}\in H^j(-\infty<\tha<\infty)$ for $0\le j\le \frac{3q+2}{2}$. As by Fourier inversion,
\be
\frac{h}{p}(k)=\int_{-\infty}^{\infty}e^{is\tha(k)}\widehat{(\frac{h}{p})}(s)\bar d s,\quad \frac{k_0}{2}<k<k_0,
\ee
where
\be
\widehat{(\frac{h}{p})}(s)=\int_{\tha(\frac{k_0}{2})}^{\tha(k_0)}e^{-is\tha(k)}\frac{h}{p}(\tha(k))\bar d \tha(k),\quad s\in \R.
\ee
where $\bar d s=\frac{ds}{\sqrt{2\pi}}$ and $\bar d \tha(k)=\frac{d\tha(k)}{\sqrt{2\pi}}$.
\par
Thus,
\be
\ba{l}
\int_{\tha(\frac{k_0}{2})}^{\tha(k_0)}\left|\left(\frac{d}{d\tha}\right)^j \frac{h}{p}(\tha(k))\right|^2|\bar d\tha(k)|\\
=\int_{\frac{k_0}{2}}^{k_0}\left|\left(\frac{4k^2k_0^2}{k_0^2-k^2}\frac{d}{dk}\right)^j \frac{h}{p}(k)\right|^2|\frac{k_0^2-k^2}{4k^2k_0^2}|\bar dk\le C<\infty,
\ea
\ee
for $0<k_0<M,0\le j\le \frac{3q+2}{2}$. Hence,
\be
\int_{-\infty}^{\infty}(1+s^2)^j|\widehat{(h/p)}(s)|^2ds \le C<\infty
\ee
for $0<k_0<M,0\le j\le \frac{3q+2}{2}$.
\par
Split
\be
\ba{rrl}
h(k)&=&p(k)\int_{t}^{\infty}e^{is\tha(k)}\widehat{(h/p)}(s)\bar ds+p(k)\int_{-\infty}^{t}e^{is\tha(k)}\widehat{(h/p)}(s)\bar ds\\
&=&h_{\Rmnum{1}}(k)+h_{\Rmnum{2}}(k).
\ea
\ee
Thus, for $\frac{k_0}{2}<k<k_0\le M$ and any positive integer $n\le \frac{3q+2}{2}$.
\be
\ba{rrl}
|e^{-2it\tha(k)}h_{\Rmnum{1}}(k)|&\le&|p(k)|\int_{t}^{\infty}|\widehat{(h/p)}(s)|\bar ds\\
&\le&|p(k)|(\int_{t}^{\infty}(1+s^2)^{-n}\bar ds)^{\frac{1}{2}}(\int_{t}^{\infty}(1+s^2)^n|\widehat{(h/p)}(s)|^2\bar ds)^{\frac{1}{2}}\\
&\le&\frac{c}{t^{n-\frac{1}{2}}}.
\ea
\ee
Consider the contour $l_1:k(u)=k_0+uk_0e^{-i\frac{3\pi}{4}},0\le u\le \frac{1}{\sqrt{2}}$. Since $\re i\tha(k)$ is positive on this contour, $h_{\Rmnum{2}}(k)$ has an analytic continuation to contours $l_1$.
\par
On the contour $l_1$,
\be
\ba{rrl}
|e^{-2it\tha(k)}h_{\Rmnum{2}}(k)|&\le&|k+k_0|^q(k_0u)^qe^{-t\re i\tha(k)}\int_{-\infty}^{t}e^{(s-t)\re i\tha(k)}\widehat{(h/p)}(s)\bar ds\\
&\le &ck_0^{2q}u^{q}e^{-t\re i\tha(k)}(\int_{-\infty}^{t}(1+s^2)^{-1}\bar ds)^{\frac{1}{2}}(\int_{-\infty}^{t}(1+s^2)|\widehat{(h/p)}(s)|^2\bar ds)^{\frac{1}{2}}\\
&\le&ck_0^{2q}u^{q}e^{-t\re i\tha(k)}.
\ea
\ee
Recall $\tha(k)=-\frac{1}{4k_0^2}(k+\frac{k_0^2}{k})$, thus
\be
\re i\tha(k)=-\im \tha=\frac{\sqrt{2}u}{8k_0}\left[\frac{1}{u^2-\sqrt{2}u+1}-1\right]\ge \frac{u^2}{4k_0}
\ee
for $0\le u\le \frac{1}{\sqrt{2}}$.
\par
Thus, on the contour $l_1$
\be
\ba{rrl}
|e^{-2it\tha(k)}h_{\Rmnum{2}}(k)|&\le &c k_0^{2q}u^{q}e^{-t\frac{u^2}{4k_0}}\le ck_0^{2q}u^{q}e^{-t\frac{u^2k_0}{4M^2}}\\
&\le& \frac{c_1}{t^{\frac{q}{2}}},
\ea
\ee
for $k_0<M$.
\par
Fix $\eps$, $0<\eps<\frac{1}{\sqrt{2}}$. If $k(u)$ is on the contour $l_1$ , $\eps<u<\frac{1}{\sqrt{2}}$, then we obtain
\be
|e^{-2it\tha(k)}R(k)|\le ce^{-\frac{u^2}{2k_0}t}\le ce^{-\frac{\eps^2}{2M}t}
\ee

\par
\subsection{\bf 2.$0<|k|<\frac{k_0}{2},k\in \R$.}
\par
We consider $0<k<\frac{k_0}{2}$, the case for $-\frac{k_0}{2}<k<0$ is similarly.
\par
Define
\be
\left\{
\ba{rrll}
\rho(\tha)&=&\rho(k(\tha)),&\tha<\tha(\frac{k_0}{2}),\\
&=&0,&\tha\ge \tha(\frac{k_0}{2}).
\ea
\right.
\ee
We claim that $\rho(\tha)\in H^j(-\infty<\tha<\infty)$ for any nonnegative integer $j$.
\par
By Fourier inversion,
\be
\rho(\tha(k))=\int_{-\infty}^{\infty}e^{is\tha(k)}\hat \rho(s)\bar ds,\quad 0<k<\frac{k_0}{2},
\ee
where
\be
\hat \rho(s)=\int_{-\infty}^{\tha(\frac{k_0}{2})}e^{-is\tha(k)}\rho(\tha(k))\bar d\tha(k).
\ee
Then,
\be
\ba{l}
\int_{-\infty}^{\tha(\frac{k_0}{2})}\left|\left(\frac{d}{d\tha}\right)^j \rho(\tha(k))\right|^2|\bar d\tha(k)|\\
=\int_{0}^{\frac{k_0}{2}}\left|\left(\frac{4k^2k_0^2}{k_0^2-k^2}\frac{d}{dk}\right)^j \rho(k)\right|^2|\frac{k_0^2-k^2}{4k^2k_0^2}|\bar dk\le C<\infty,
\ea
\ee
for any nonnegative integer $j$, $0<k_0<M$, since $r(k)\rightarrow 0$ rapidly, as $k\rightarrow 0$.
\par
Hence
\be
\int_{-\infty}^{\infty}(1+s^2)^j |\hat \rho(s)|^2\bar ds\le C,
\ee
for any nonnegative integer $j$.
\par
Split
\be
\ba{rrl}
\rho(k)&=&\int_{t}^{\infty}e^{is\tha(k)}\hat \rho(s)\bar ds+\int_{-\infty}^{t}e^{is\tha(k)}\hat \rho(s)\bar ds\\
&=&h_{\Rmnum{1}}(k)+h_{\Rmnum{2}}(k).
\ea
\ee
Then, for $0<k<\frac{k_0}{2}$ and any positive integer $j$, we obtain,
\be
\ba{rrl}
|e^{-2it\tha(k)}h_{\Rmnum{1}}(k)|&\le& \int_{t}^{\infty}|\hat \rho|\bar ds\\
&\le& (\int_{t}^{\infty}(1+s^2)^{-j}\bar ds)^{\frac{1}{2}}(\int_{t}^{\infty}(1+s^2)^j|\hat \rho(s)|^2\bar ds)^{\frac{1}{2}}\\
&\le& \frac{c}{t^{j-\frac{1}{2}}}.
\ea
\ee
Consider the contour $l_2: k(u)=uk_0e^{-i\frac{\pi}{4}},0<u<\frac{1}{\sqrt{2}}$. Since $\re i\tha(k)$ is positive on this contour, $h_{\Rmnum{2}}$ has an analytic
continuation to contour $l_2$.
\par
On the contour $l_2$,
\be
\ba{rrl}
|e^{-2it\tha(k)}h_{\Rmnum{2}}(k)|&\le&e^{-t\re i\tha(k)}\int_{-\infty}^{t}e^{(s-t)\re i\tha(k)}|\hat \rho(k)|\bar ds\\
&\le&e^{-t\re i\tha(k)}(\int_{-\infty}^{t}(1+s^2)^{-1}\bar ds)^{\frac{1}{2}}(\int_{-\infty}^{t}(1+s^2)|\hat \rho(k)|^2\bar ds)^{\frac{1}{2}},
\ea
\ee
where
\be
\re i\tha(k)=\frac{\sqrt{2}}{8k_0}(\frac{1}{u}-u)\ge \frac{1}{8k_0},
\ee
for $0<u\le \frac{1}{\sqrt{2}}$.
\par
Thus, we obtain,
\be
|e^{-2it\tha(k)}h_{\Rmnum{2}}(k)|\le ce^{-\frac{t}{8k_0}}.
\ee

\par
\subsection{\bf 3. $|k|>k_0,k\in \R$}
\par
We consider $k>k_0$, the case for $k<-k_0$ is similarly.
\par
Set
\be
\rho(k)=\frac{\ol{r(\bar k)}}{1+|r(k)|^2}.
\ee
We write
\be
(k+i)^{m+5}\rho(k)=\mu_0+\mu_1(k-k_0)+\cdots+\mu_m(k-k_0)^m+\frac{1}{m!}\int_{k_0}^{k}((\cdot+i)^{m+5}\rho(\cdot))^{(m+1)}(\gam)(k-\gam)^md\gam.
\ee
Define
\be
R(k)=\frac{\sum_{i=0}^{m}\mu_i(k-k_0)^i}{(k+i)^{m+5}}
\ee
and write $\rho(k)=h(k)+R(k)$. We have
\be
\left.\frac{d^jh(k)}{dk^j}\right|_{k_0}=0,\quad 0\le j\le m.
\ee
For $0<k_0<M$, set
\be
v(k)=\frac{(k-k_0)^q}{(k+i)^{q+2}}.
\ee
Let
\be
\left\{
\ba{rrll}
\frac{h}{v}(\tha)&=&\frac{h}{v}(k(\tha)),&\tha<\tha(k_0),\\
&=&0,&\tha\ge \tha(k_0).
\ea
\right.
\ee
Then
\be
\frac{h}{v}(\tha(k))=\int_{-\infty}^{\infty}e^{is\tha(k)}\widehat{(\frac{h}{v})}(s)\bar ds,\quad k\ge k_0,
\ee
where
\be
\widehat{(\frac{h}{v})}(s)=\int_{-\infty}^{\tha(k_0)}e^{-is\tha(k)}\frac{h}{v}(\tha(k))\bar d\tha(k).
\ee
Moreover, we have
\be
\frac{h}{v}(\tha(k))=\frac{(k-k_0)^{3q+2}}{(k+i)^{3q+4}}g(k,k_0),
\ee
where
\be
g(k,k_0)=\frac{1}{m!}\int_{0}^{1}((\cdot+i)^{m+5}\rho(\cdot))^{(m+1)}(k_0+u(k-k_0))(1-u)^k du
\ee
and
\be
\left|\frac{d^j g(k,k_0)}{dk^j}\right|\le C,\quad k\ge k_0.
\ee
Using the identity $\left|\frac{k-k_0}{k+k_0}\right|\le 1$ for $k\ge k_0$, we have
\be
\ba{rrl}
\int_{-\infty}^{\tha(k_0)}\left|\left(\frac{d}{d\tha}\right)^j\left(\frac{h}{v}(\tha(k))\right)\right|^2\bar d\tha(k)&=&
\int_{k_0}^{\infty}\left|\left(\frac{4k^2k_0^2}{k_0^2-k^2}\frac{d}{dk}\right)^j\frac{h}{v}(k)\right|^2|\frac{k_0^2-k^2}{4k_0^2k^2}|\bar dk\\
&\le &c\int_{k_0}^{\infty}\left|\frac{(k-k_0)^{3q+2-3j}}{(k-i)^{3q+4}}\right|^2k^{4j-2}(k^2-k_0^2)\bar dk \le C_1, \quad 0\le j\le \frac{3q+2}{3}.
\ea
\ee
Thus,
\be
\int_{-\infty}^{\infty}(1+s^2)^j \left|\widehat{(\frac{h}{v})}(s)\right|^2\bar ds\le C<\infty.
\ee
We write
\be
\ba{rrl}
h(k)&=&v(k)\int_{t}^{\infty}e^{is\tha(k)}\widehat{(h/v)}(s)\bar ds+v(k)\int_{-\infty}^{t}e^{is\tha(k)}\widehat{(h/v)}(s)\bar ds\\
&=&h_{\Rmnum{1}}(k)+h_{\Rmnum{2}}(k).
\ea
\ee
For $k\ge k_0,0<k_0<M$, and any positive integer $e\le \frac{3q+2}{3}$, we obtain,
\be
\ba{rrl}
|e^{-2it\tha(k)}h_{\Rmnum{1}}(k)|&\le& \frac{|k-k_0|^q}{|k+i|^{q+2}}\int_{t}^{\infty}|\widehat{(h/v)}(s)|\bar ds\\
&\le & \frac{|k-k_0|^q}{|k+i|^{q+2}}(\int_{t}^{\infty}(1+s^2)^{-e}\bar ds)^{\frac{1}{2}}(\int_{t}^{\infty}(1+s^2)^e|\widehat{(h/v)}(s)|^2\bar ds)^{\frac{1}{2}}\\
&\le & c\frac{1}{(1+|k|^2)t^{e-\frac{1}{2}}}.
\ea
\ee
And $h_{\Rmnum{2}}(k)$ has an analytic continuation to the upper half-plane, where $\re i\tha(k)$ is positive. We estimate $e^{-2it\tha(k)}h_{\Rmnum{2}}(k)$ on the contour $k(u)=k_0+uk_0e^{i\frac{\pi}{4}},u\ge 0$.
\par
If $0<u\le M_1$,
\be
|e^{-2it\tha(k)}h_{\Rmnum{2}}(k)|\le c\frac{k_0^qu^qe^{-t\re i\tha(k)}}{|k-i|^{q+2}},
\ee
where
\be
\re i\tha(k)=\frac{\sqrt{2}u}{8k_0}\frac{u^2+\sqrt{2}u}{u^2+\sqrt{2}u+1}\ge \frac{1}{4k_0}\frac{u^2}{u^2+\sqrt{2}u+1}
\ee
Then
\be
\ba{rrl}
|e^{-2it\tha(k)}h_{\Rmnum{2}}(k)|&\le &c\frac{k_0^qu^qe^{-t\re i\tha(k)}}{|k+i|^{q+2}}\le c_1\frac{k_0^qu^q}{|k+i|^{q+2}}e^{-t\frac{1}{4k_0}\frac{u^2}{u^2+\sqrt{2}u+1}}\\
&\le &\frac{c_2}{(1+|k|^2)^{q+2}t^{\frac{q}{2}}}\le \frac{c_3}{(1+|k|^2)t^{\frac{q}{2}}}
\ea
\ee
If $u>M_1$, then
\be
\re i\tha(k)\ge \frac{\sqrt{2}}{8k_0}\frac{u^3}{u^2+\sqrt{2}u+1}
\ee
and
\be
\ba{rrl}
|e^{-2it\tha(k)}h_{\Rmnum{2}}(k)|&\le &c\frac{k_0^qu^qe^{-t\re i\tha(k)}}{|k+i|^{q+2}}\le c_4\frac{k_0^qu^q}{|k+i|^{q+2}}e^{-t\frac{\sqrt{2}}{8k_0}\frac{u^3}{u^2+\sqrt{2}u+1}}\\
&\le &\frac{c_5}{(1+|k|^2)t^{q}}
\ea
\ee
Hence, for $u>0$, we obtain
\be
|e^{-2it\tha(k)}h_{\Rmnum{2}}(k)|\le \frac{c_5}{(1+|k|^2)t^{\frac{q}{2}}}.
\ee
Note that if $l$ is an arbitrary positive integer, we can choose $m$ large enough such that $\frac{3q+2}{2}-\frac{1}{2}>q-\frac{1}{2}>\frac{q}{2}>l$, so the proposition holds.

\section{How to reduce the contour $\Sig$ to $\Sig^{(3)}$}\label{appedb}
Let us now show that, in the sense of appropriately defined operator norms, one may
always choose to delete (or add) a portion of a contour(s) on which the jump is $\id$,
without altering the Riemann-Hilbert problem in the operator sense.
\par
Suppose that $\Sig_{1}$ and $\Sig_{2}$ are two oriented skeletons in $\C$ with
\be
\mbox{card}(\Sig_{1}\cap \Sig_{2})<\infty;
\ee
let $u=u(\lam)=u_+(\lam)+u_-(\lam)$ be a $2\times 2$ matrix-valued function on
\be
\Sig_{12}=\Sig_{1}\cup \Sig_{2}
\ee
with entries in $L^2(\Sig_{12})\cap L^{\infty}(\Sig_{12})$ and suppose that
\be
u=0\qquad \qquad \mbox{on }\Sig_{2}.
\ee
Let
\be
R_{\Sig_{1}}\mbox{ denote the restriction map }L^2(\Sig_{12})\rightarrow L^2(\Sig_{1}),
\ee
\be
\id_{\Sig_{1}\rightarrow \Sig^{(12)}}\mbox{ denote the embedding }L^2(\Sig_{1})\rightarrow L^2(\Sig_{12}),
\ee
\be
C_{u}^{12}:L^2(\Sig_{12})\rightarrow L^2(\Sig_{12})\mbox{ denote the operator in (\ref{BCRHPCom}) with }u\leftrightarrow \om,
\ee
\be
C_u^1:L^2(\Sig_{1})\rightarrow L^2(\Sig_{1})\mbox{ denote the operator in (\ref{BCRHPCom}) with }u \uparrow \Sig_{1}\leftrightarrow \om,
\ee
\be
C_u^E:L^2(\Sig_{1})\rightarrow L^2(\Sig_{12})\mbox{ denote the restriction of $C_u^{12}$ to }L^2(\Sig_{1}).
\ee
And, finally, let
\be
\left\{
\ba{l}
\id_{\Sig_{1}}\mbox{ and }\id_{\Sig_{12}}\mbox{ denote the identity operators on}\\
L^2(\Sig_{1})\mbox{ and }L^2(\Sig_{12}),\mbox{ respectively}.
\ea
\right.
\ee
We then have the next lemma:
\begin{lemma} \label{opralema}
\be \label{ideopra1}
C_u^{12}C_u^E=C_u^EC_u^{12},
\ee
\be \label{ideopra2}
(\id_{\Sig_{1}}-C_u^1)^{-1}=R_{\Sig_{1}}(\id_{\Sig_{12}}-C_u^{12})^{-1}\id_{\Sig_{1}\rightarrow \Sig_{12}},
\ee
\be \label{ideopra3}
(\id_{\Sig_{12}}-C_u^{12})^{-1}=\id_{\Sig_{12}}+C_u^E(\id_{\Sig_{1}}-C_u^1)^{-1}R_{\Sig_{1}},
\ee
in the sense that if the right-hand side of (\ref{ideopra2}),resp. (\ref{ideopra3}), exists,
then the left-hand side exists and identity (\ref{ideopra2}),resp. (\ref{ideopra3}), holds true.
\end{lemma}
\begin{proof}
See Lemma 2.56 in \cite{dz}.
\end{proof}


\section{Proof of the proposition \ref{tasyinftyrsol}}\label{appdd}
\begin{proof}
Write
\be
 \ba{l}
R\left(\frac{k}{\sqrt{k_0^{-3}t}}+k_0\right)(\dta^1_A(k))^{2}-R(k_0\pm)k^{-2i\nu}e^{i\frac{k^2}{2}}
\\=
e^{i\frac{\beta}{2}k^2}e^{i\frac{\beta}{2}k^2}R\left(\frac{k}{\sqrt{k_0^{-3}t}}+k_0\right)k^{-2i\nu}e^{i(1-2\beta)\frac{k^2}{2}(1-\frac{k}{(1-2\beta)\eta^4\sqrt{k_0^{-9}t} })}\\
{}\times \left(\frac{2k_0}{2k_0+\frac{k}{\sqrt{k_0^{-3}t}}}\right)^{-2i\nu}
{} e^{2\left(\chi(\frac{k}{\sqrt{k_0^{-3}t}}+k_0)-\chi (k_0)\right)}
{}-e^{i\frac{\beta}{2}k^2}e^{i\frac{\beta}{2}k^2}R(k_0\pm)k^{-2i\nu}e^{i(1-2\beta)\frac{k^2}{2}}
 \ea
\ee
 and also divided it into six terms
 \be
   R\left(\frac{k}{\sqrt{k_0^{-3}t}}+k_0\right)(\dta^1_A(k))^{2}-R(k_0\pm)k^{-2i\nu}e^{i\frac{k^2}{2}}
   =e^{i\beta\frac{k^2}{2}}(\Rmnum{1}+\Rmnum{2}+\Rmnum{3}+\Rmnum{4})
 \ee
where
\[
\ba{l}
\Rmnum{1}=e^{i(1-\beta)\frac{k^2}{2}}k^{-2i\nu}[R(\frac{k}{\sqrt{k_0^{-3}t}}+k_0)-R(k_0\pm)]\\
\Rmnum{2}=e^{i\beta\frac{k^2}{2}}k^{-2i\nu}R(\frac{k}{\sqrt{k_0^{-3}t}}+k_0)
   \left(e^{i(1-2\beta)\frac{k^2}{2}(1-\frac{k}{(1-2\beta)\eta^4\sqrt{k_0^{-9}t} })}-e^{i(1-2\beta)\frac{k^2}{2}}\right)\\
\Rmnum{3}=e^{i\beta\frac{k^2}{2}}k^{-2i\nu}R(\frac{k}{\sqrt{k_0^{-3}t}}+k_0)
   e^{i(1-2\beta)\frac{k^2}{2}(1-\frac{k}{(1-2\beta)\eta^4\sqrt{k_0^{-9}t} })}
   \left(\left(\frac{2k_0}{2k_0+\frac{k}{\sqrt{k_0^{-3}t}}}\right)^{-2i\nu}-1\right)\\
\Rmnum{4}=e^{i\beta\frac{k^2}{2}}k^{-2i\nu}R(\frac{k}{\sqrt{k_0^{-3}t}}+k_0)
   e^{i(1-2\beta)\frac{k^2}{2}(1-\frac{k}{(1-2\beta)\eta^4\sqrt{k_0^{-9}t} })}
   \left(\frac{2k_0}{2k_0+\frac{k}{\sqrt{k_0^{-3}t}}}\right)^{-2i\nu}\\
   {}\left(e^{2\left(\chi(\frac{k}{\sqrt{k_0^{-3}t}}+k_0)-\chi (k_0)\right)}-1\right)\\
\ea
\]
Note that $|e^{i\beta\frac{k^2}{2}}|=e^{-\frac{\beta \lam^2k_0^2}{2}}$ and $|k^{-2i\nu}|=e^{2\nu\arg{k}}\le C$, where $C$ is a constant which is independent of $k$, for $k=k_0\lam e^{\frac{i\pi}{4}},-\eps<\lam<\eps$. The terms $\Rmnum{1},\Rmnum{2},\Rmnum{3}$ 
and $\Rmnum{4}$ can be estimated as follows.
\[
\ba{rrl}
|\Rmnum{1}|&\le &|k^{-2i\nu}|\cdot |e^{i\beta\frac{k^2}{2}}|\cdot |\frac{k}{\sqrt{k_0^{-3}t}}|\cdot ||\partial_{k}R(k)||_{L^{\infty}}\\
&\le & \frac{C}{\sqrt{t}},
\ea
\]

\[
\ba{rrl}
|\Rmnum{2}|&\le &|k^{-2i\nu}|\cdot |e^{i\beta\frac{k^2}{2}}|\cdot ||R||_{L^{\infty}}\cdot \left|\frac{d}{ds}e^{i(1-2\beta)\frac{k^2}{2}(1-s\frac{k}{(1-2\beta)\eta^4\sqrt{k_0^{-9}t} })}\right|,\quad 0<s<1\\
&\le &\frac{C}{\sqrt{t}}
\ea
\]
To estimate $\Rmnum{3}$, we write
\[
\ba{rrl}
|\Rmnum{3}|&\le & |k^{-2i\nu}|\cdot |e^{i\beta\frac{k^2}{2}}|\cdot ||R||_{L^{\infty}}\cdot |e^{i(1-2\beta)\frac{k^2}{2}(1-\frac{k}{(1-2\beta)\eta^4\sqrt{k_0^{-9}t} })}|\cdot |\int_{1}^{1+\sqrt{\frac{k_0}{4t}}k}2i\nu \xi^{2i\nu-1}d\xi|\\
&\le & \frac{C}{\sqrt{t}},
\ea
\]
as $|\xi^{2i\nu-1}|\le c$ for $\xi=1+s\sqrt{\frac{k_0}{4t}}k,0\le s\le 1.$

\par
The estimate for $\Rmnum{4}$ is as follows,
\[
\ba{rrl}
|\Rmnum{4}|&\le & C\sup_{0\le s\le 1}|e^{2\left(\chi(\frac{k}{\sqrt{k_0^{-3}t}}+k_0)-\chi (k_0)\right)}|\cdot \left|2e^{i\beta\frac{k^2}{2}}(\chi(\frac{k}{\sqrt{k_0^{-3}t}}+k_0)-\chi (k_0))\right|
\ea
\]
now let us show how to control $\left|e^{i\beta\frac{k^2}{2}}(\chi(\frac{k}{\sqrt{k_0^{-3}t}}+k_0)-\chi (k_0))\right|$.
\[
\ba{l}
\left|e^{i\beta\frac{k^2}{2}}(\chi(\frac{k}{\sqrt{k_0^{-3}t}}+k_0)-\chi (k_0))\right|\\
=\left|\frac{e^{i\beta\frac{k^2}{2}}}{2\pi}\int_{-\infty}^{-k_0}\ln{\frac{k_0-s+\frac{k}{\sqrt{k_0^{-3}t}}}{k_0-s}}d\ln{(1+|r(s)|^2)}+\int_{k_0}^{\infty}\ln{\frac{s-k_0-\frac{k}{\sqrt{k_0^{-3}t}}}{s-k_0}}d\ln{(1+|r(s)|^2)}\right|\\
={}\left|\frac{e^{i\beta\frac{k^2}{2}}}{2\pi}\int_{-\infty}^{-1}\ln{(1+\frac{\sqrt{\frac{k_0}{t}}k}{1-s})}d\ln{(1+|r(sk_0)|^2)}+\int_{1}^{\infty}\ln{(1-\frac{\sqrt{\frac{k_0}{t}}k}{s-1})}d\ln{(1+|r(sk_0)|^2)}\right|\\
=|\Rmnum{4}_1+\Rmnum{4}_2|
\ea
\]
where
\[
\Rmnum{4}_1=\frac{e^{i\beta\frac{k^2}{2}}}{2\pi}\int_{-\infty}^{-1}(g(s)-g(1))\ln{(1+\frac{\sqrt{\frac{k_0}{t}}k}{1-s})}ds+\int_{1}^{\infty}(g(s)-g(1))\ln{(1-\frac{\sqrt{\frac{k_0}{t}}k}{s-1})}ds
\]
\[
\Rmnum{4}_2=\frac{e^{i\beta\frac{k^2}{2}}}{2\pi}\int_{-\infty}^{-1}g(1)\ln{(1+\frac{\sqrt{\frac{k_0}{t}}k}{1-s})}ds+\int_{1}^{\infty}g(1)\ln{(1-\frac{\sqrt{\frac{k_0}{t}}k}{s-1})}ds
\]
here $g(s)=\partial_{s}\ln{(1+|r(sk_0)|^2)}$. 
\par
Then, using the Lipschitz condition $|\ln{(1+a)}|\le |a|$, we have
\[
\ba{l}
|\Rmnum{4}_1|\le \left|\frac{e^{i\beta\frac{k^2}{2}}}{2\pi}\right|\cdot \int_{-\infty}^{-1}\sqrt{\frac{k_0}{t}}|k|\cdot |\frac{g(s)-g(1)}{s-1}|ds
+\left|\frac{e^{i\beta\frac{k^2}{2}}}{2\pi}\right|\cdot \int_{1}^{\infty}\sqrt{\frac{k_0}{t}}|k|\cdot |\frac{g(s)-g(1)}{s-1}|ds\\
{}\le C t^{-1/2}.
\ea
\]
Notice that $g(s)$ is rapidly decay as $s\rightarrow \infty$, so the above integral is well-defined. And notice $|ke^{i\beta\frac{k^2}{2}}|$ and $\partial_{s}g$ are bounded.

\[
\ba{l}
|\Rmnum{4}_2|= \left|\frac{e^{i\beta\frac{k^2}{2}}}{2\pi}\right|\cdot \left|\int_{2}^{\infty}{\ln{(1-\frac{k_0k^2}{s^2t})}}ds\right|\\
{}= \left|\frac{e^{i\beta\frac{k^2}{2}}}{2\pi}\right|\cdot \left|(\int_{2}^{L}+\int_{L}^{\infty}){\ln{(1-\frac{k_0k^2}{s^2t})}}ds\right|,
\ea
\]
where $L$ is a big-enough positive constant. Since the above infinity integral is well-defined, the second integral is very small. And integral by parts shows that the first integral becomes
\[
\int_{2}^{L}{\ln{(1-\frac{k_0k^2}{s^2t})}}ds=\left.\left(\ln{(s-\sqrt{\frac{k_0}{t}}k)}(s-\sqrt{\frac{k_0}{t}}k)+\ln{(s+\sqrt{\frac{k_0}{t}}k)}(s+\sqrt{\frac{k_0}{t}}k)-2s\ln{s}\right)\right|_{2}^{L}
\]
so we have
\[
\ba{rrl}
|\Rmnum{4}_2|&\le &C\frac{\log t}{\sqrt{t}}.
\ea
\]
Then, we can get the estimate of (\ref{Rasyt}). This finishes the proof of the proposition \ref{tasyinftyrsol}.
\end{proof}

\end{document}